\documentclass[11pt]{article}
\usepackage{macro}
\usepackage{float}
\usepackage[normalem]{ulem}
\usepackage{mwe}
\usepackage{soul}
\usepackage{hyperref}
\usepackage[dvipsnames]{xcolor}
\hypersetup{colorlinks = true, linkcolor=black, citecolor=MidnightBlue, urlcolor=MidnightBlue}

\def\bea{\begin{eqnarray}}
\def\eea{\end{eqnarray}}
\def\>{\rangle}
\def\<{\langle}
\def\rs{\mathrm{s}}

\newcommand{\oprnorm}[1]{\left\lVert #1 \right\rVert}
\newcommand{\oprnormsm}[1]{\lVert #1 \rVert}
\renewcommand{\leq}{\le}
\renewcommand{\geq}{\ge}

\newtheorem{taggedtheoremx}{Lemma}
\newenvironment{taggedlemma}[1]
 {\renewcommand\thetaggedtheoremx{#1}\taggedtheoremx}
 {\endtaggedtheoremx}

\title{
No-go theorems for sequential preparation of two-dimensional chiral states via channel-state correspondence
}

\author[1]{Ruihua Fan}
\author[2]{Yantao Wu}
\author[3]{Yimu Bao}
\author[1,4]{Zhehao Dai}
\affil[1]{\normalsize\it Department of Physics, University of California, Berkeley, CA 94720, USA}
\affil[2]{\normalsize\it Institute of Physics, Chinese Academy of Sciences, Beijing 100190, China}
\affil[3]{\normalsize\it Kavli Institute for Theoretical Physics, University of California, Santa Barbara, CA 93106, USA}
\affil[4]{\normalsize\it Department of Physics and Astronomy, University of Pittsburgh, PA 15213, USA}
\begin{document}

\maketitle	

\begin{abstract}
We investigate whether sequential unitary circuits can prepare two-dimensional chiral states, using a correspondence between sequentially prepared states, isometric tensor network states, and one-dimensional quantum channel circuits.
We establish two no-go theorems, one for Gaussian fermion systems and one for generic interacting systems.
In Gaussian fermion systems, the correspondence relates the defining features of chiral wave functions in their entanglement spectrum to the algebraic decaying correlations in the steady state of channel dynamics.
We establish the no-go theorem by proving that local channel dynamics with translational invariance cannot support such correlations. 
As a direct implication, two-dimensional Gaussian fermion isometric tensor network states cannot support algebraically decaying correlations in all directions or represent a chiral state.
In generic interacting systems, we establish a no-go theorem by showing that the state prepared by sequential circuits cannot host the tripartite entanglement of a chiral state due to the constraints from causality.
\end{abstract}

\tableofcontents

\section{Introduction}
\label{sec:intro}

Efficient preparation of complex quantum states on near-term quantum devices is a central topic in quantum many-body physics~\cite{Semeghini:2021wls,Andersen:2022xmz,Iqbal:2023wvm,Kim:2023bwr,Bluvstein:2023zmt,Altman:2019vbv}. 
A key step toward this goal is to identify the fundamental limitations of different preparation protocols.
Among the proposed architectures, sequential unitary circuits stand out as a compelling approach~\cite{Cirac:sequential2005,cirac:2dsequential2008}.
On the one hand, each qubit is acted on by only a finite number of gates, keeping the overall error rate low.
On the other hand, the circuit has an extensive depth and can generate long-range entanglement~\cite{Lieb:1972wy}.
Indeed, it has been shown that sequential unitary circuits can prepare highly entangled quantum states, such as the string-net states in two dimensions and fractons in three dimensions~\cite{FuDaoZhiDa:2020vsq,PRXQuantum.3.040315,Chen:2023qst,Satzinger:2021eqy}. 

Given their power, we ask whether sequential circuits can also prepare two-dimensional (2D) chiral states, topological states that have a nonzero chiral central charge~\cite{KaneFisher1997,Cappelli:2001mp,Kitaev:2005hzj}.
The fact that such states do not admit fixed-point wave functions with strict area-law entanglement makes this question particularly interesting~\cite{Kapustin:nogo,Kim:2024amo,Li:2024iwn}.
Previous work has provided explicit preparation protocols in two scenarios: parallel wires that are continuous in one spatial direction and discrete in the other, and lattice systems where the state is prepared through sequential adiabatic evolution generated by gapped Hamiltonians~\cite{Chen:2024fvu}.
However, the question remains open for the most experimentally relevant settings: can a 2D chiral state be prepared on a lattice using a sequential circuit built from local unitary gates with compact support? 

We tackle this question by relating the capability of sequential circuits to another, more elementary subject: the steady states of local quantum channel dynamics in one lower dimension.
For example, consider a two-dimensional array of qudits on an $L_x \times L_y$ grid, initialized in a trivial product state in the bulk. Each step of the sequential circuit acts only locally in the $y$-direction
\begin{equation}
	\ket{\psi(t)} \mapsto \ket{\psi(t+1)} = \prod_{x=1}^{L_x} U_{x,y=t} \ket{\psi(t)}
 \label{eq:sequential circuit}
\end{equation}
where $U_{x,y}$ is a local unitary acting in the vicinity of $(x,y)$. For simplicity, we require that $U_{x,y}$'s with the same $y$ commute so that we can apply them in parallel.
As we detail later, this two-dimensional unitary dynamics can be recast as a one-dimensional Markovian channel circuit evolving a mixed state
\begin{equation}
	\rho(t) \mapsto \rho(t+1) = \prod_{x=1}^{L_x} \calN_{x,t}[\rho(t)]\,,
	\label{eq:local channel circuit}
\end{equation}
where $\calN_{x,t}$ is a local channel that corresponds to $U_{x,y=t}$ in \eqnref{eq:sequential circuit}. By construction, each $\calN_{x,t}$ is compactly supported near $x$, and the $\calN_{x,t}$'s with the same time $t$ commute. \figref{fig:correspondence channel isotns}~(a) and (b) show one pair of such examples.
Importantly, the spectrum of the mixed state $\rho(T)$ coincides with the entanglement spectrum of the output state of the sequential circuit for the subregion $y < T$.  
In two dimensions, the entanglement structure is a useful fingerprint for all topological properties~\cite{Kitaev:2005dm,LevinWen:2006,Li:2008kda,Shi:2019mlt}. Therefore, analyzing the steady state of channel dynamics helps determine the topology of the output of the sequential circuit, hence its capabilities.

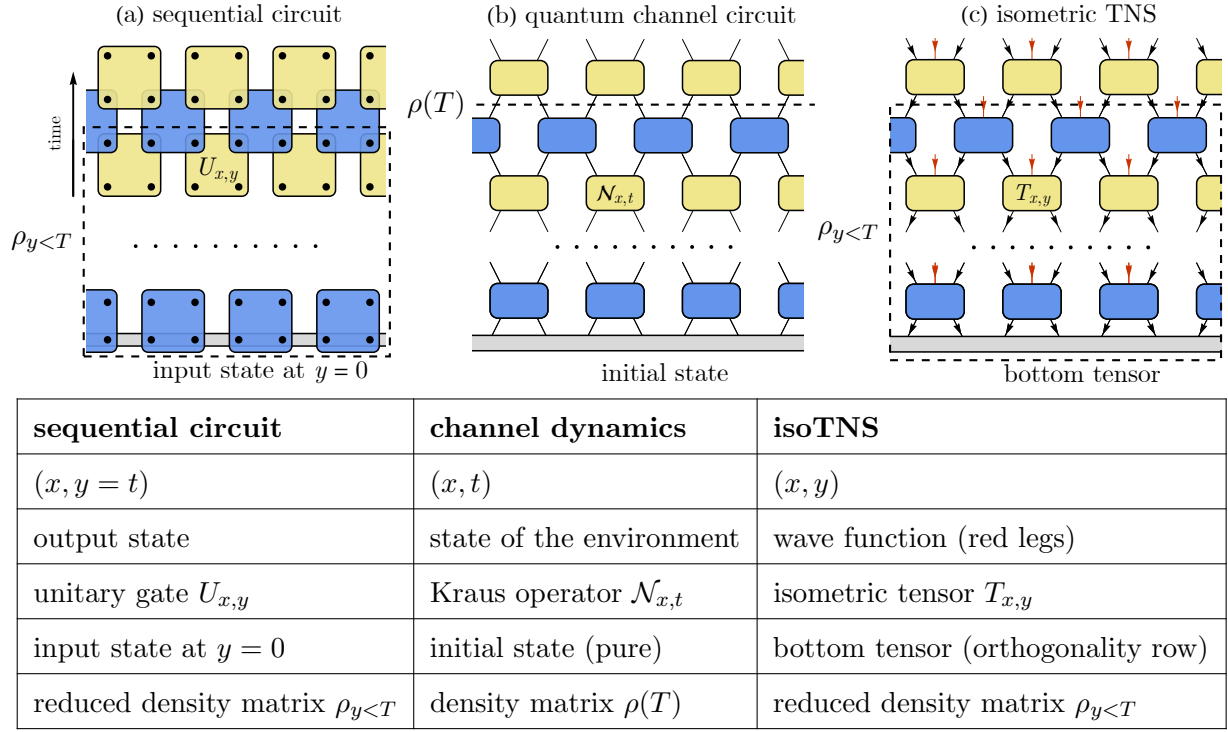
\begin{figure}[t]
\centering
\begin{tikzpicture}
\node at (0,0) {\includegraphics[width=0.92\textwidth]{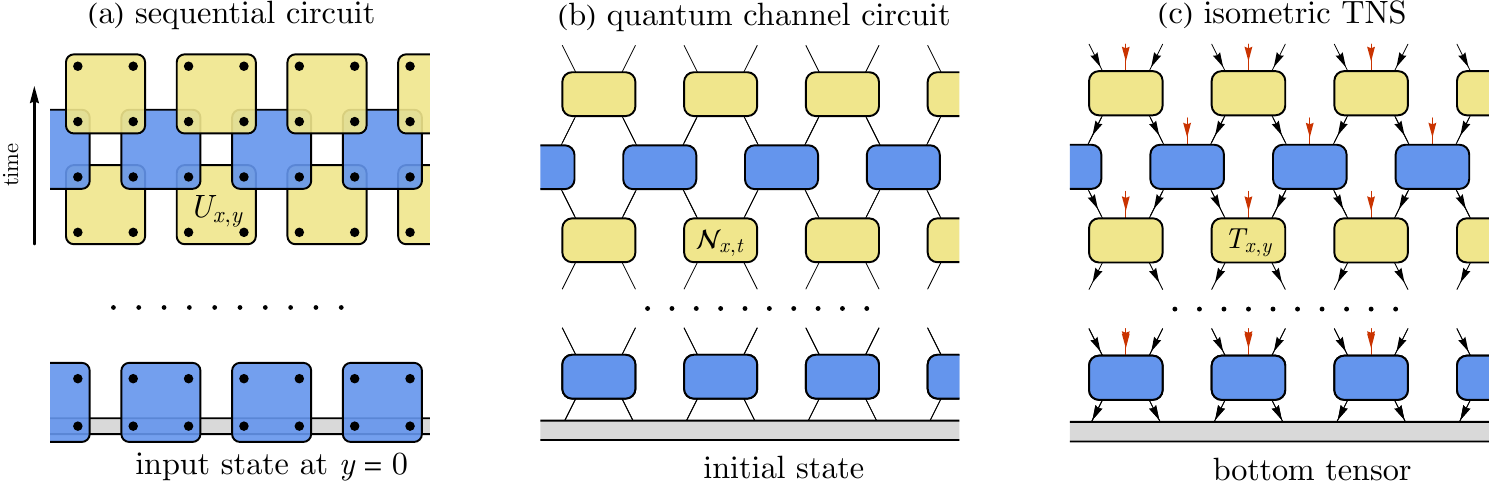}};
\draw[black,dashed,thick] (-7.32,-2.15) -- (-3.25,-2.15) -- (-3.25,0.9) -- (-7.32,0.9) -- node[left]{$\rho_{y<T}$} cycle;
\draw[black,dashed,thick] (-2.1,1.2) node[left] {$\rho(T)$} -- (2.4,1.2);
\draw[black,dashed,thick] (3.4,-2.2) -- (7.8,-2.2) -- (7.8,1.2) -- (3.4,1.2) -- node[left]{$\rho_{y<T}$} cycle;
\end{tikzpicture}
\vspace{10pt}
\renewcommand{\arraystretch}{1.5}
    \begin{tabular}{|l|l|l|}
        \hline
        \textbf{sequential circuit} & \textbf{channel dynamics} & \textbf{isoTNS} \\
        \hline
        $(x,y=t)$ & $(x,t)$ & $(x,y)$ \\ 
        \hline
        output state & state of the environment & wave function (red legs) \\
        \hline 
        unitary gate $U_{x,y}$ & Kraus operator $\mathcal{N}_{x,t}$ & isometric tensor $T_{x,y}$ \\
        \hline
        input state at $y=0$ & initial state (pure) & bottom tensor (orthogonality row) \\ 
        \hline
        reduced density matrix $\rho_{y<T}$ & density matrix $\rho(T)$ & reduced density matrix $\rho_{y<T}$ \\ 
        \hline
\end{tabular}
\caption{Correspondence between (a) 1+1D brick-wall channel circuit, (b) 2D isometric tensor network states and (c) sequential unitary circuits. }
\label{fig:correspondence channel isotns}
\end{figure}

This channel-state correspondence is best understood in the language of isometric tensor network states (isoTNS)~\cite{Cirac:sequential2005,Zaletel:2020roc}, a subclass of tensor network states (TNS) where each tensor satisfies an isometry constraint. 
Specifically, the isoTNS wave function can represent all possible states prepared by sequential circuits, such as that in \eqnref{eq:sequential circuit}, while also capturing the full history of the local channel dynamics in \eqnref{eq:local channel circuit}.
\figref{fig:correspondence channel isotns}~(c) shows the explicit form of the isoTNS.
The isometric constraint was originally introduced to simplify numerical algorithms in two dimensions; here, it naturally corresponds to the unitarity of sequential circuits or the trace-preserving condition in quantum channel dynamics. 
The capability of the sequential circuit is equivalent to the representability of such an isoTNS ansatz. 
In this context, it was proved that 2D Gaussian fermion TNS can represent chiral states at the cost of having power-law correlation functions in the bulk~\cite{Dubail:2013pda}. Our question amounts to asking whether this is still possible with the isometric constraint. 

Our answer is summarized by two no-go theorems, one of which is proved for Gaussian fermion circuits and the other for generic interacting circuits on the cylindrical geometry. Let the initial state of the sequential circuit be a 2D product state except that the bottom row is an entangled 1D state---the resource. Then we have
\begin{rmk}[No-go theorem 1]
    Given an arbitrary Gaussian resource state, any 2D translationally invariant Gaussian fermion sequential circuit cannot prepare a chiral state in the thermodynamic limit, or equivalently, a Gaussian isoTNS cannot represent a chiral state.
\end{rmk}
\begin{rmk}[No-go theorem 2]
    If the resource is a finite-bond-dimensional matrix product state, any sequential 2D circuit with 2-local gates cannot generate gapless edges, described by conformal field theories, on an $L_x\times L_y$ cylinder for $L_y < L_x/6$. Here, the cylinder is open in $y$ direction.
\end{rmk}

In what follows, we outline our proof strategy and the structure of the remainder of the manuscript. 
In \secref{sec:correspondence}, we review the channel-state correspondence and explain the case for Gaussian fermions in detail.
For Gaussian wave functions, the chirality hinges on the presence of chiral modes in the single-particle entanglement spectrum, which depends on the steady-state correlation functions of the associated Gaussian quantum channel, according to the correspondence.
In \secref{sec:free fermion no go}, we prove an important intermediate result: for any translationally invariant local Gaussian channel, all single-particle modes that are coupled to the environment can at most mediate exponentially decaying spatial correlations in the long-time limit.
This result, combined with the correspondence, leads to no-go theorem 1.
Our theorem also extends to Gaussian sequential circuits of $\calO(L^2)$ depth.
For a generic wave function of interacting systems, the presence of chiral edge modes is diagnosed by tripartite entanglement.
In \secref{sec:interacting}, we show that the tripartite entanglement is constrained by the causality of sequential circuits, leading to no-go theorem 2.
We conclude with the discussion in \secref{sec:discussion}.

\section{Channel-state correspondence}
\label{sec:correspondence}

In this section, we review the correspondence between the $d$-dimensional local quantum channel, the $d+1$-dimensional sequential quantum circuit, and $d+1$-dimensional isoTNS.
The correspondence is well-known for $d = 0$ in the tensor network community~\cite{Cirac:2020obd} and has been formulated in Ref.~\cite{Zaletel:2020roc,Malz:2024val} for $d \geq 1$.
The correspondence to the channel dynamics helps us address the limitation of sequential circuits in state preparation later in this work.
In the rest of this paper, for convenience, we often use the language of isoTNS instead of the terminology of sequential circuits.

\subsection{Correspondence in the qudit systems}
\label{subsec:generalcorrespondence}

We lay out the correspondence in spin systems in dimension $d = 1$; generalization to other dimensions is straightforward~\cite{Cirac:sequential2005,3DisoTNS}. 
We start with the mapping between the channel dynamics and isoTNS and then explain its equivalence to isoTNS and sequential circuits. 
The table in \figref{fig:correspondence channel isotns} summarizes the dictionary.

We consider a one-dimensional quantum channel circuit acting on an $L$-qudit chain with a periodic boundary condition.
At each time step, the density matrix $\rho(t)$ evolves under local channels as
\begin{equation}
    \rho(t+1) = \calN_t [\rho(t)] = \prod_x \calN_{x,t}[\rho(t)]\,,
\end{equation}
where $(x,t)$ denotes the spacetime coordinate of each channel, and we assume that the channels $\calN_{x,t}$ mutually commute for the same $t$.
For simplicity, let $\calN_{x,t}$ act only on two nearest-neighbor qudits and be arranged in a brick-wall pattern, shown in \figref{fig:correspondence channel isotns}~(a).
The circuit describes a discretized Markovian open system dynamics.
Each channel can describe local unitary evolution, local dissipation, or measurement with local feedback. 
The dynamics still satisfies the Lieb-Robinson bound~\cite{Poulin:LRbound2010} but not necessarily the quantum detailed balance conditions~\cite{Kossakowski:1977db,FU07,Temme:2010big}.

The 1D channel circuit can be mapped to a 2D isoTNS by considering the Stinespring dilation for each channel.
Specifically, we purify each local channel $\calN_{x,t}$ as a unitary $U_{x,t}$ that acts on the two system qudits $\calH_x\otimes \calH_{x+1}$ and an environment qudit $\calH^{\text{env}}_x$. 
Without loss of generality, we fix the initial state of the environment as $\ket{0}$ so that the unitary becomes an isometric map $T_{x,t}$,
\begin{equation}
\begin{gathered}
	\calN_{x,t} \mapsto T_{x,t} = \begin{gathered} \includegraphics[width=0.9cm]{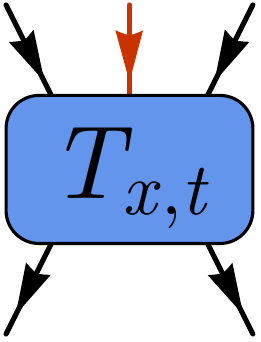} \end{gathered} \\
    T_{x,t}:\calH_x\otimes \calH_{x+1} \rightarrow \calH_x\otimes \calH_{x+1} \otimes \calH^{\text{env}}_x,
\end{gathered}
\end{equation}
where the black legs denote the two system qudits and the red leg in the middle denotes the environment qudit.
Arrows on the legs indicate the isometric condition, following the convention in Ref.~\cite{Zaletel:2020roc}, which is \textit{opposite} to the time direction.
Accordingly, the purified 1+1D channel circuit prepares a 2D isoTNS on environment qudits as shown in \figref{fig:correspondence channel isotns}~(c).
We refer to the $x$ direction of the isoTNS as the spatial direction, and the $y$ direction as the temporal direction.
Each purified channel gives rise to the local isometric tensor.
Each system qudit on the 1D chain becomes the virtual leg of the isoTNS, whereas each environment qudit becomes a physical leg.
The initial state of the channel dynamics specifies the bottom boundary condition of the isoTNS, i.e. the quantum state of the virtual legs at the bottom. 
The bottom row of this isoTNS is commonly referred to as the orthogonality row.
To avoid confusion, we stick to the words ``virtual" and ``physical" when referring to the degrees of freedom.

The isometric tensors in the network shown in \figref{fig:correspondence channel isotns}~(c) are arranged sequentially from the bottom to the top, which can be naturally reorganized into a linear-depth sequential circuit shown in \figref{fig:correspondence channel isotns}~(a).
Specifically, we map each isometric tensor into a unitary gate acting on four qudits
\footnote{Unitary gates on four qudits map to isoTNS tensors with virtual dimension $d$ and physical dimension $d^2$. To map an isoTNS tensor with general physical and virtual dimensions to a unitary gate on 4 qudits, we need to choose a large enough $d$ and enlarge the dimension of virtual (physical) legs to $d$ ($d^2$).}
\begin{equation}\label{Eq:Lshape}
\begin{gathered}
	\includegraphics[width = 1.5cm]{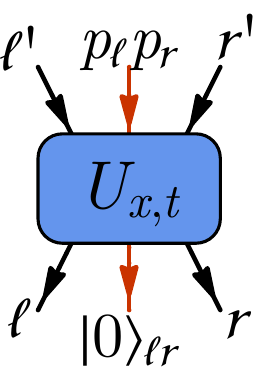}
\end{gathered}
\,=\,
\begin{gathered}
	\includegraphics[width = 1.6cm]{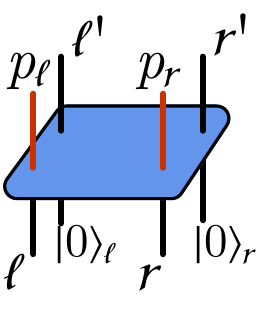}
\end{gathered}\,,
\end{equation}
where the splitting of the Hilbert space for the physical leg is not unique.
The sequential circuit shares the same time direction as the 1+1D channel dynamics.
The dynamics starts with a one-dimensional input state (the quantum state of the bottom virtual legs of the isoTNS).
At each step, the circuit entangles the system with an extra array of qudits in the product states and eventually prepares a 2D state. 
The equivalence between states prepared by sequential circuits of various other architectures and 2D tensor network states has been formulated in Ref.~\cite{CiracSequential2022}.

We close this section with a simple lemma on the relation between the density matrix in the channel dynamics and the entanglement spectrum of the sequential unitary circuit's output state.
\begin{lemma}
Consider a 2D sequential circuit with the input of an arbitrary pure 1D state on the bottom row.
The reduced density matrix of the prepared 2D state below a horizontal cut at $y$ is isospectral to the 1D density matrix of the corresponding channel dynamics at time $t_0=y$. (Note that in order to get a pure physical state without any post selection, we turn the virtual legs into physical legs at the top edge)
\end{lemma}

\begin{proof}
Let $\ket{\psi_0}_V$ be the initial state of the 1D channel circuit. Let $A$ and $B$ be the physical legs in the lower and upper half-plane separated by the time slice at $t_0$, $U_1$ and $U_2$ the concatenation of all unitaries applied before and after $t_0$ in the circuit.
The sequential circuit up to time $t_0$ prepares an entangled state between $A$ and the 1D array of virtual legs $V$:
\begin{equation}
    \ket{\psi}_{A,V} = U_1 \ket{0}_A \otimes \ket{\psi_0}_V = \sum_i \sqrt{\lambda_i} \ket{i}_A \otimes \ket{i}_V
\end{equation}
where we used the Schmidt decomposition in the second step. The virtual reduced density matrix $\rho_V = \sum_i \lambda_i \ket{i} \bra{i}_V$ is the density matrix of the system in the channel dynamics. It has the same spectrum as the reduced density matrix $\rho_A$.
The unitary gates after $t_0$, $U_2$, act only on the virtual legs and the upper half of the physical system and therefore do not change $\rho_A$. This leads to the isospectral relation $\rho_A\sim\rho_V$. 
Since the final state of the sequential circuit is a pure state of $A$ and $B$, we have $\rho_B\sim\rho_A\sim\rho_V$.
\end{proof}

The correspondence rules out the preparation of chiral states using sequential circuits whose corresponding 1D quantum channel exhibits rapid mixing.
Consider preparing chiral states on a cylinder with an open boundary condition in the $y$ direction.
Such chiral states must have gapless edge modes that mediate power-law correlations on the top and bottom edges.
An initial state may be chosen to mimic the correlation of the bottom edge, but if the corresponding channel reaches a unique steady state at a time $t$ sublinear in the system size, the output of the sequential circuit at $y> t$ will be independent of the initial state, hence the same as that prepared from a trivial product state. In \secref{sec:interacting} we prove that sequential circuits are unable to generate the required entanglement from a trivial product state. Therefore, sequential circuits whose corresponding quantum channel reaches a unique steady state in sublinear time cannot prepare a chiral state from any 1D input state.

We note that the above argument does not rule out the possibility of preparing chiral states on a torus, i.e., maintaining the entanglement between the bottom and the top edge throughout the evolution.
This approach requires the channel to preserve information during the evolution and is adopted by Ref.~\cite{Chen:2024fvu}.
Yet, in Sec.~\ref{sec:free fermion no go}, we prove a no-go theorem for Gaussian fermions independent of the boundary conditions and with no assumption on the equilibrium behaviors of the corresponding Gaussian channel.

\subsection{Correspondence in Gaussian fermion systems}\label{subsec:GaussianChannelKitaevChain}

In this section, we detail the correspondence in Gaussian fermion systems.
We consider the Gaussian quantum channel acting on a 0D system and formulate the state of the environment as a Gaussian matrix product state (MPS).
Results in this section apply to translation-invariant Gaussian fermion channels in higher dimensions, e.g., those studied in Sec.~\ref{sec:free fermion no go}; there, correlation in each momentum sector corresponds to that in a 0D Gaussian channel dynamics.
Drawing on tools from~\cite{Bravyifermionlinearoptics}, we characterize the steady-state correlation of a Gaussian quantum channel (\thmref{thm:steadystate}).
We further discuss its implications on the boundary properties, bulk reduced density matrix and correlation length of the Gaussian MPS (\thmref{thm:initialstateindependence} and \ref{thm:correlationlength}).
The key difference between Gaussian fermion MPS and general MPS is the lack of cat-state-like correlation, which allows separating bulk properties from boundary contributions.
We note that these results are of independent interest, besides their application to the proof of the no-go theorem in Sec.~\ref{sec:free fermion no go}.

Consider a system of $2n$ Majorana fermions $c_1,\dots,c_{2n}$ with the anti-commutation relation $\{c_j, c_l\} = 2 \delta_{jl}$. A Gaussian state $\rho$ of this system obeys Wick's theorem and can be fully characterized by the two-point correlation matrix
\begin{equation}
    \Gamma_{jl} \equiv \frac{i}{2} \Tr \big( \rho[c_j,c_l] \big)\,.
\end{equation}
The $2n\times 2n$ matrix $\Gamma$ is real antisymmetric, and $\Gamma\Gamma^T\le I$ with the equality satisfied for pure states.

A generic Gaussian fermion quantum channel $\calN$ preserves Gaussianity of quantum states, namely the output state remains fully characterized by the correlation matrix.
The channel acts as a linear map on the correlation matrix~\cite{Bravyifermionlinearoptics}
\begin{equation}
\label{eq:gaussian channel linear equation}
    \Gamma \mapsto \calN[\Gamma] = A + B \Gamma B^T\,,
\end{equation}
where $A$ is real anti-symmetric, $B$ is real.
The channel $\calN$ being a completely positive trace preserving (CPTP) map requires $A,B$ to satisfy the inequality,
\begin{equation}
    \Lambda^T \Lambda \leq I\,,\quad \Lambda \equiv \begin{pmatrix} A & B \\ -B^T & 0 \end{pmatrix}\,.
    \label{eq:CPTP condition}
\end{equation}
In particular, we have $B^T B \leq I$, namely the spectral norm (the maximum singular value of $B$) is less than or equal to 1, and $B^T B = I$ if the channel is an isometric embedding.

The Gaussian channel has an extra structure if the input and output Hilbert spaces are the same.
In this case, $B$ is a square matrix with the same dimension as $A$. The subspace that has unit-norm eigenvalues of $B$ is preserved under repeated actions of the channel, while the remaining degrees of freedom are dissipative with a biased noise described by $A$.
This intuition is formulated precisely in the lemma below.
\begin{lemma}
\label{lemma:akbkeigenvec}
Let $\ket{v}$ be a right eigenvector of $B$ with a unit-norm eigenvalue $b$, i.e. $B\ket{v} = b \ket{v}$. 
Then, the eigenvalue $b$ does not have a nontrivial Jordan block, $|v\>$ is orthogonal to other generalized eigenvectors of $B$, and $A|v\> = 0$.
\end{lemma}

\begin{proof}
The proof only uses the CPTP condition in \eqnref{eq:CPTP condition}.
By definition, we have $\<v|B^T B|v\> = 1$.
Since $B^T B\leq I$, the equation above can be satisfied only when $B^T B |v\> = |v\>$. Thus we have
\begin{equation}
    B^T |v\> = b^*|v\>\ \Rightarrow\ \<v|B = b\<v|
\end{equation}
Namely, $\bra{v}$ is a left eigenvector of eigenvalue $b$. 
Let $\ket{u}$ be an arbitrary vector orthogonal to $\ket{v}$,
we have $\<v|B|u\> = b\<v|u\> = 0$, i.e. $B$ acts within the orthogonal complement of $|v\>$. This implies there cannot be non-trivial Jordan blocks for unit-norm eigenvalues, and $|v\>$ is orthogonal to other generalized eigenvectors of $B$.
Using $A^T A + B B^T\leq I$, we have
\begin{equation}
    \braket{v|A^T A + B B^T|v} \leq 1 \ \Rightarrow\ 
    \braket{v|A^T A|v} = 0.
\end{equation}
Since $A$ is real anti-symmetric, we have $A|v\> = 0$ and $\<v|A = 0$.
\end{proof}

Lemma~\ref{lemma:akbkeigenvec} allows us to decompose the Hilbert space of single-particle modes into the mutually orthogonal preserved $V_u$ and dissipative modes $V_d$, i.e. $V = V_u \oplus V_d$.
In particular, $V_u = \mathrm{span}\{|v\> \,|\, B|v\> = b|v\>, |b| = 1\}$ is spanned by the eigenvectors of $B$ associated with unit-norm eigenvalues, and $V_d$ is the orthogonal complement of $V_u$.

Lemma~\ref{lemma:akbkeigenvec} also has an implication on the dynamics of dissipative and preserved modes. The fact that $V_u$ and $V_d$ are invariant subspaces of $A$ and $B$ implies $A = (0, 0; 0, A_d)$ and $B = (U, 0; 0, B_d)$, where $U$ is real orthogonal (therefore, unitary) and any eigenvalue of $B_d$ has a norm less than 1. 
If we repeatedly apply the same channel $\calN$, the preserved modes undergo unitary evolution, and the dissipative modes converge exponentially in time to a unique steady state. 

\begin{thm}\label{thm:steadystate}
Consider a Gaussian fermion channel $\calN$ in Eq.~\eqref{eq:gaussian channel linear equation} acting on the space of $2n$ Majorana fermion modes $V = V_u\oplus V_d$.
After applying the channel $t$ times, the correlation matrix approaches a block diagonal form
\begin{align}
\oprnorm{\calN^t[\Gamma^{(0)}] - \Gamma^{(\rs)}} \leq e^{-t\ln(1/r)+\calO(\ln t)}, \quad \Gamma^{(\rs)} \equiv \begin{pmatrix} U^t\Gamma_{u}^{(0)}(U^\dag)^t & 0\\ 0&  \Gamma_{d}^{(\rs)}\end{pmatrix}
\label{Eq:steadystatecorrelationmatrix},
\end{align}
where $\Gamma_{u}^{(0)}$ is the correlation matrix of the initial state in the preserved subspace $V_u$, $\Gamma_{d}^{(\rs)}$ is the steady-state correlation matrix in the dissipative subspace $V_d$ that depends only on the channel $\calN$, and $r<1$ is the spectral radius of $B_d$.
\end{thm}
\begin{proof}
Applying the channel $\calN$ to a Gaussian state $t$ times gives rise to
\begin{align}
    \mathcal{N}^t[\Gamma_0] = A{(t)} + B^t\Gamma_0 (B^T)^{t}\,,\quad A{(t)}\equiv \sum_{s = 0}^{t-1}B^sA(B^T)^s\,.
\end{align}
In terms of the preserved and dissipative modes, the correlation matrix takes the form
\begin{align}
    \calN^t\left[\Gamma^{(0)}\right] = \begin{pmatrix} U^t \Gamma_{u}^{(0)} (U^\dagger)^t & U^t \Gamma_{ud}^{(0)}(B_d^T)^t \\ -B_d^t(\Gamma_{ud}^{(0)})^T(U^\dagger)^t & B_d^t \Gamma_{d}^{(0)} (B_d^T)^t + \sum_{s = 0}^{t-1} B_d^s A_d (B_d^T)^s
    \end{pmatrix}
\end{align}
where $\Gamma^{(0)} = (\Gamma_{u}^{(0)}, \Gamma_{ud}^{(0)}; -\Gamma_{ud}^{(0),T}, \Gamma_{d}^{(0)})$ is the initial correlation matrix.

The evolved Gaussian state converges to a block-diagonal correlation matrix $\Gamma^{(\rs)}$ in Eq.~\eqref{Eq:steadystatecorrelationmatrix} as $t$ increases
\begin{align}
\oprnorm{\calN^t\left[\Gamma^{(0)}\right] - \Gamma^{(\rs)}} &\leq 2 \oprnorm{U^t \Gamma_{ud}^{(0)} (B_d^T)^t} + \oprnorm{B_d^t \Gamma_{d}^{(0)} (B_d^T)^t} + \sum_{s = t}^{\infty} \oprnorm{B_d^s A_d (B_d^T)^s} \nonumber \\
&\leq 2e^{-t\ln(1/r)+\calO(\ln t)} + e^{-2t\ln(1/r)+\calO(\ln t)} + e^{-2t\ln(1/r)+\calO(\ln t)} = e^{-t\ln(1/r)+\calO(\ln t)},
\end{align}
where $\Gamma_{d}^{(\rs)} = \sum_{s = 0}^\infty B_d^s A_d (B_d^T)^s$, $r<1$ is the spectral radius of $B_d$ (the largest norm of its eigenvalues).
We use the upper bound on the spectral norms, $\oprnormsm{\Gamma_{ud}^{(0)}}$, $\oprnormsm{\Gamma_{d}^{(0)}}$, $\oprnormsm{A_d} \le 1$, and $\lVert B_d^t \rVert \le e^{-t\ln(1/r)+\calO(\ln t)}$~\footnote{To bound the spectral norm $\lVert B_d^t \rVert$, we express $B_d = Q J Q^{-1}$ in terms of its Jordan normal form $J$ and a similarity transformation $Q$.
The Jordan form $J = D + N$, where $D$ is the diagonal part given by the eigenvalues of $B$, and $N$ is a nilpotent matrix with only non-vanishing entries on the first off-diagonal taking unit values (i.e. $N^{\dim(B_d)} = 0$). 
The spectral norm then has an upper bound, for $t \ge \dim(B_d)-1$,
\begin{align}
\lVert B_d^t \rVert &\leq \oprnorm{Q}\lVert J^t \rVert\oprnorm{Q^{-1}} \leq \oprnorm{Q}\oprnorm{Q^{-1}}\sum_{s = 0}^{t}\binom{t}{s} \lVert D^{t-s}\rVert \cdot \lVert N^{s} \rVert = \oprnorm{Q}\oprnorm{Q^{-1}} \sum_{s = 0}^{\dim(B_d)-1} \binom{t}{s} \lVert D^{t-s}\rVert \\
&\leq \oprnorm{Q}\oprnorm{Q^{-1}} \; t^{2n} \; r^{t-2n} \leq e^{-t\ln(1/r)+\calO(\ln t)}, \label{eq:spectral_norm_matrix_power}
\end{align}
where $r < 1$ is the spectral radius of $B_d$ (i.e. the norm of the largest eigenvalue), and $2n \ge \dim(B_d)$ is the total number of Majorana modes.}.
\end{proof}

\begin{figure}
\centering
\begin{tikzpicture}
\node at (0,0) {\includegraphics[width=0.35\textwidth]{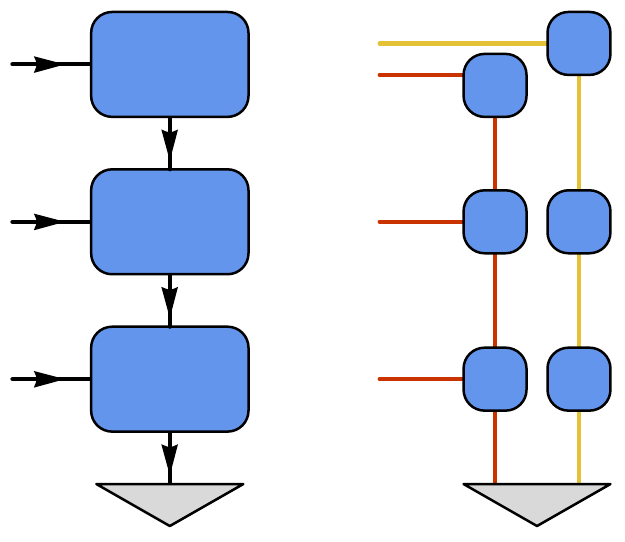}};    
\node[scale=2] at (0.1,0) {$=$};
\node[] at (-1.35,0.4) {$\calN_{V\mapsto VP}$};
\node[] at (-2.8,-2.3) {initial state};
\node[] at (0.7,-1.8) {dissipative};
\node[] at (3.5,-1.8) {preserved};
\end{tikzpicture}

\caption{Illustration of Gaussian fermion matrix product state. We interpret each tensor (blue box) as an isometric map from its bottom virtual leg to its physical and top virtual leg. This isometric map corresponds to a quantum channel from the bottom virtual leg to the top virtual leg of each tensor. Red (yellow) lines represent dissipative (preserved) modes of the quantum channel. Only the dissipative modes (red lines) are coupled to the physical legs in the bulk. The boundary condition of the lowest tensor (green triangle) corresponds to the initial state of the quantum channel.}
\label{Fig:GFMPS}
\end{figure}
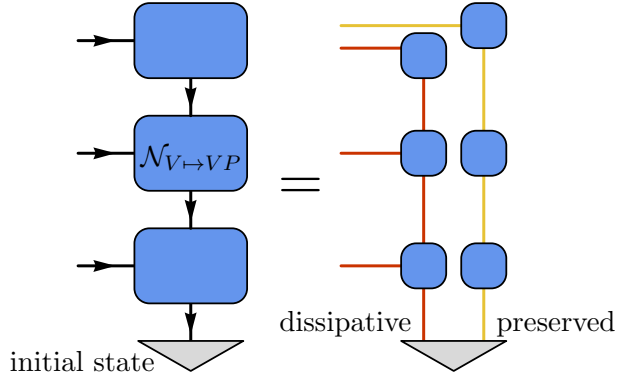

By lifting each channel $\calN$ to an isometric map acting on an enlarged system with auxiliary degrees of freedom, we obtain the correspondence between the Gaussian fermion channel dynamics in 0D and a 1D Gaussian fermion MPS.
Each MPS tensor represents the isometric map from the bottom virtual leg to the top virtual leg and the physical leg (as shown in \figref{Fig:GFMPS}), which induces a linear map on the correlation matrix,
\begin{equation}
    \calN_{V \mapsto VP}[\Gamma_V] = \begin{pmatrix} A_{P} & A_{PV}  \\ -A_{PV}^T  & A \end{pmatrix} + 
\begin{pmatrix} B_{P}  \\ B \end{pmatrix}
\Gamma_{V} 
\begin{pmatrix} B_{P}^T  & B^T \end{pmatrix}
\end{equation}
The channel being isometric requires
\begin{equation}
\label{eq:tensorcorrelationmatrix}
    \Lambda^T\Lambda=I\,,\quad \text{where }\Lambda = \begin{pmatrix} A_{P} & A_{PV} & B_{P} \\ -A_{PV}^T & A & B\\
    -B_{P}^T & -B^T & 0
    \end{pmatrix}\,.
\end{equation}
We use $D$ to fully specify the MPS tensor, and the relation between $A,B$ and $D$ to specify the channel-state correspondence for Gaussian fermion systems (see Appendix~\ref{app:GfTNS} and Ref.~\cite{wu2025alternatinggaussianfermionicisometric,PhysRevA.81.052338,PhysRevB.107.125128} for an alternative interpretation of $D$).

We note in passing that the correspondence relates the parity of preserved Majorana modes in Gaussian channel dynamics to the topological index of the MPS.
As a concrete example, the ground states of the 1D Majorana chain in the topological and trivial phases correspond to Gaussian channels with odd and even numbers of preserved modes, respectively (see \appref{app:KitaevChain}).

In the rest of this section, we discuss the implications of \thmref{thm:steadystate} on the bulk and boundary properties of the Gaussian MPS, in particular, its correlation functions and entanglement spectrum.
\begin{thm}\label{thm:initialstateindependence}
Consider a Gaussian fermion MPS in the isometric form. 
The correlation matrix $\Gamma_P$ of a physical subsystem that is distance $y$ from the bottom boundary is given by
\begin{align}
    \oprnorm{\Gamma_P - \Gamma_{P,\mathrm{bulk}}} \leq e^{-y\ln(1/r)+\calO(\ln y)},
\end{align}
where $0<r<1$ is the spectral radius of $B_d$, and $\Gamma_{P,\text{bulk}}$ is a correlation matrix that depends on the channel associated with the MPS tensor but not on the bottom boundary condition.
\end{thm}

\begin{proof}
The Gaussian MPS in the subregion from $y$ to  $y+ \ell$ defines an isometric map $\tilde{\calN}_{V \mapsto PV}$ from the bottom virtual leg to the physical legs and the top virtual leg
\begin{align}
    \tilde{\calN}_{V \mapsto PV}[\Gamma_V] = \begin{pmatrix} \tilde A_{P} & \tilde{A}_{PV}  \\ \tilde{A}_{VP}  & A{(\ell)} \end{pmatrix} + 
\begin{pmatrix} \tilde{B}_{P}  \\ B^\ell \end{pmatrix}
\Gamma_{V} 
\begin{pmatrix} \tilde{B}_{P}^T  & (B^T)^\ell \end{pmatrix} .
\end{align}

According to \thmref{thm:steadystate}, the correlation matrix at the bottom virtual bond is given by
\begin{align}
    \oprnorm{\Gamma_V - \Gamma^{(\rs)}} \leq e^{-y\ln(1/r)+\calO(\ln y)}, \quad \Gamma^{(\rs)} = \begin{pmatrix} U^{y}\,\Gamma_{u}^{(0)}\,(U^\dag)^{y} & 0\\ 0&  \Gamma_{d}^{(\rs)}\end{pmatrix}
\end{align}
where $0< r < 1$ is the spectral radius of $B_d$.

Since $B^\ell$ acts as a real orthogonal matrix in the preserved subspace, and $(B^T)^\ell B^\ell + \tilde{B}_P^T \tilde{B}_P = I$, we have $\tilde{B}_P = (0, \tilde{B}_{P,d})$ acting trivially in the preserved subspace.
Hence, the correlation matrix on the physical legs $\Gamma_P = \tilde{A}_{P} + \tilde{B}_P \Gamma_V \tilde{B}_P^T$ can be approximated by $\Gamma_{P,\text{bulk}} \equiv \tilde{A}_{P} + \tilde{B}_{P,d} \Gamma_{d}^{(\rs)} \tilde{B}_{P,d}^T$, which is independent of the state at the bottom boundary, i.e.
\begin{align}
\oprnorm{\Gamma_P - \Gamma_{P,\text{bulk}}} = \oprnorm{\tilde{B}_P \left(\Gamma_V - \Gamma^{(\rs)}\right)\tilde{B}_P^T} \leq e^{-y\ln(1/r)+\calO(\ln y)}.
\end{align}
\end{proof}

The spectral norm of $\Gamma_P - \Gamma_{P,\mathrm{bulk}}$ sets an upper bound on each matrix element, i.e., the difference between a two-point correlation function in two Gaussian states.
Theorem~\ref{thm:initialstateindependence} suggests that, away from the bottom boundary, the physical correlation function is fully determined by the MPS tensor. 
We note that this is a special property for Gaussian fermion MPS. In general cases, only injective MPS fully specifies the correlation function from individual tensors.
In a non-injective MPS representing a cat state, changing the boundary condition can convert it into a trivial product state~\cite{Cirac:2020obd}. 
In what follows, we further determine the correlation length from the MPS tensor.

\begin{thm}\label{thm:correlationlength}
    For a Gaussian fermion matrix product state corresponding to the channel in Eq.~\eqref{eq:gaussian channel linear equation}, the correlation length is $\xi \leq 1/\ln(1/r)$, where $r$ is the spectral radius of $B_d$.
\end{thm}
\begin{proof}
The correlation matrix on physical site $i$ and $j$ is given by
\begin{align}
    \Gamma_{ij} &= \calN_{V \mapsto P} \circ \calN^{|j-i|-1} \circ \calN_{V \mapsto VP}[\Gamma^{(\rs)}] \nonumber \\
    &= \begin{pmatrix}
        A_P + B_P\Gamma^{(\rs)} B_P^T & (A_{PV}+B_P \Gamma^{(\rs)} B^T)(B^T)^{|j-i|-1} B_P^T \\
        B_P B^{|j-i|-1} (-A_{PV}^T - B\Gamma^{(\rs)} B_P^T) & A_P + B_P\left(A{(|j-i|)}+B^{|j-i|}\Gamma^{(\rs)}(B^T)^{|j-i|}\right)B_P^T
    \end{pmatrix}
\end{align}
The two-point correlation function between any two Majorana modes $c_{i,\alpha}$ and $c_{j,\beta}$ on sites $i$ and $j$ (matrix element of the off-diagonal block) has an upper bound
\begin{align}
    |\langle i c_{i,\alpha}c_{j,\beta}\rangle| &\le \left\Vert (A_{PV}+B_P \Gamma^{(\rs)} B^T)(B^T)^{|j-i|-1} B_P^T\right\Vert = \left\Vert (A_{PV}+B_P \Gamma^{(\rs)} B^T)\begin{pmatrix}
        0 \\
        (B_d^T)^{|j-i|-1}B_{P,d}^T
    \end{pmatrix}\right\Vert \nonumber \\
    &\le e^{-|j-i|\ln(1/r)+\calO(\ln |j-i|)}.
\end{align}
Here, we use the spectral norm $\oprnormsm{B_d^{|j-i|}}$ in Eq.~\eqref{eq:spectral_norm_matrix_power} and the facts that $A_{PV}$, $B_P$, and $\Gamma^{(\rs)}$ all have a spectral norm that is less than one. 
This gives an upper bound on the correlation length $\xi \leq 1/\ln(1/r)$.
\end{proof}

The correspondence has further implications on the entanglement spectrum of the Gaussian MPS in the bulk, which plays an important role in establishing the no-go theorem in Sec.~\ref{sec:free fermion no go}.
Unlike correlation functions, the entanglement spectrum in the bulk can depend on the bottom boundary condition. 
This is manifest from the channel perspective.
Starting from a given initial state of the Gaussian channel (the boundary vector in \figref{Fig:GFMPS}), information in the dissipative modes (red lines in \figref{Fig:GFMPS}) leaks into the physical legs, determining the physical state near the bottom boundary.
In contrast, the information in the preserved modes (yellow lines in \figref{Fig:GFMPS}) remains in the virtual legs; at the top boundary, these virtual legs become physical degrees of freedom.
If the initial state entangled the dissipative and the preserved modes, the sequentially prepared state would exhibit a finite amount of entanglement between the bottom and top degrees of freedom even in the thermodynamic limit, although the top and the bottom are not coupled directly throughout the preparation~\cite{Chen:2024fvu}.
In this case, the entanglement spectrum for the upper half system receives contributions also from the entanglement between the bottom and top edges.
This motivates the following definition that separates the bulk properties from edge contributions.
\begin{defn}[Bulk entanglement spectrum]
Consider the steady-state correlation matrix in the dissipative subspace $\Gamma_{d}^{(\rs)}$. The bulk entanglement spectrum is defined as $\{\varepsilon_{\alpha,\mathrm{bulk}} = \log\frac{1+\lambda_\alpha}{1-\lambda_\alpha}\}$, where $\{\pm i\lambda_\alpha\}$ with $\lambda_\alpha \ge 0$ are eigenvalues of $\Gamma_{d}^{(\rs)}$.
\end{defn}

The bulk entanglement spectrum defined above is independent of boundary conditions and forms a subset of the full single-particle entanglement spectrum, as is elaborated in Sec.~\ref{sec:free fermion no go}.
For any sequentially prepared state, one can always construct a state with identical bulk correlation, and its entanglement spectrum is given by $\{\varepsilon_{\alpha,\mathrm{bulk}}\}$.
When proving no-go results in Sec.~\ref{subsec:nogo_sequential}, we focus on the bulk entanglement spectrum of states that can be prepared by sequential circuits.

\section{No-go theorem for Gaussian fermion sequential circuits}
\label{sec:free fermion no go}

For Gaussian fermion states, the defining feature of chirality is the presence of chiral modes in the single-particle entanglement spectrum, which resembles the energy spectrum on the physical edge~\cite{Kitaev:2005hzj,Fidkowski:2009entanglement,Dubail:2011fpa}.
In this section, we establish the no-go theorem in two steps: we first show that the entanglement spectrum of any state prepared by a staircase sequential circuit is non-chiral (Sec.~\ref{subsec:nogo_sequential}), and then extend the argument to the output states of more general sequential circuit architectures (Sec.~\ref{sec:othercircuits}).

\subsection{No-go theorem for the staircase sequential circuit}
\label{subsec:nogo_sequential}
We consider the staircase sequential circuit in \figref{fig:correspondence channel isotns}~(a) that corresponds to the brick-wall Gaussian fermion channel circuit in \figref{fig:correspondence channel isotns}~(b).
Here, we assume that each site hosts the same number of Majorana modes, the circuit exhibits translational invariance by one unit cell (i.e. two sites) in space and two steps in time.
We prove that the 2D state prepared by such circuits cannot represent a chiral state.
As the key step, we show that the steady state of the channel circuit is analytic in the momentum space (up to removable singularities), indicating the corresponding sequentially prepared 2D state has a non-chiral entanglement spectrum.

The output state $\rho$ of a translationally invariant channel circuit can be fully characterized by the correlation matrix in the momentum space,
\begin{align}
\Gamma_{\alpha\beta,k} \equiv \frac{i}{2} \Tr(\rho[c_{\alpha,k}, c_{\beta,-k}])\,,\quad 
c_{\alpha,k} \equiv \frac{1}{\sqrt{L}} \sum_{x} e^{-i k x}c_{\alpha,x},
\end{align}
where $x$ denotes the positions of unit cells, $\alpha$, $\beta = 1,2,\cdots, 2n$ label the Majorana mode within each unit cell, and $\{c_{\alpha,k},c_{\beta,-k}\} = \delta_{\alpha,\beta}$.
We note that each mode $c_{\alpha,k}$ for $0< k <\pi$ can be regarded as a complex fermion mode, and $c_{\alpha,-k} = c_{\alpha,k}^\dagger$. 
In what follows, we focus on momenta $0< k <\pi$ unless explicitly specified.
The correlation matrix in the momentum space is not necessarily real and satisfies
\begin{equation}
    \Gamma^\dag_k = -\Gamma_k\,,\quad
    \Gamma^*_k = \Gamma_{-k}
    \,,\quad \Gamma_k \Gamma^\dag_k \le I\,.
\end{equation}
The channel circuit of two steps acts on $\Gamma_k$ as a linear map similar to \eqnref{eq:gaussian channel linear equation},
\begin{equation}
    \calN[\Gamma_k] = A_k + B_k\Gamma_kB_k^\dag.
    \label{Eq:oneperiodchannel}
\end{equation}
Here, $A_k$ and $B_k$ are analytic functions in $k$, which depend on the channel in the circuit\footnote{We here explicitly write down the functions $A_k$ and $B_k$ in Eq.~\eqref{Eq:oneperiodchannel}. 
We denote the two-site local channel in the circuit as $\mathsf{N}_{i,i+1}[\Gamma] = \mathsf{A}_{i,i+1} + \mathsf{B}_{i,i+1}\Gamma\mathsf{B}_{i,i+1}^T$.
Here, we assume the two-site channels are identical and use the subscript to label the sites they act on.
The channels associated with the odd and the even time steps of the circuits take the form $\mathcal{N}_{\text{odd(even)}}[\Gamma] = A_{\text{odd(even)}} + B_{\text{odd(even)}}\Gamma B_{\text{odd(even)}}^T$, where $A_{\text{odd}} = \oplus_{i=0}^{L/2} \mathsf{A}_{2i+1,2i+2}$, $B_{\text{odd}} = \oplus_{i=0}^{L/2} \mathsf{B}_{2i+1,2i+2}$, $A_{\text{even}} = \oplus_{i=0}^{L/2} \mathsf{A}_{2i,2i+1}$, and $B_{\text{even}} = \oplus_{i=0}^{L/2} \mathsf{B}_{2i,2i+1}$.

In the momentum space, the correlation matrix $\Gamma_k$ evolves as
\begin{align}
    \calN_{k,\text{odd}}[\Gamma_k] = \sfA + \sfB \Gamma_k \sfB^T, \quad
    \calN_{k,\text{even}}[\Gamma_k] = U_k \sfA U_k^\dag + U_k \sfB U_k^\dag \Gamma_k U_k \sfB^T U_k^\dag, \quad U_k = \begin{pmatrix} 0 & I_n \\ e^{-ik}I_n & 0 \end{pmatrix}.
\end{align}
Here, $\Gamma_k$ is a $2n \times 2n$ matrix with $n$ being the number of Majorana modes on each site, as each unit cell contains two sites.
The channel associated with two time steps is the composition $\calN_k = \calN_{k,\text{even}}\circ \calN_{k,\text{odd}}$, and the associated $A_k$ and $B_k$ are 
\begin{equation}
\label{eqn:AkBk}
A_k = U_k \sfA U_k^\dag + U_k \sfB U_k^\dag \sfA U_k \sfB^T U_k^\dag,\quad B_k = U_k \sfB U_k^\dag \sfB.
\end{equation}
}.
The channel being a CPTP map requires that $A_k^\dag A_k+B_k B_k^\dag \le I$, and $B_k^\dag B_k \le I$.
For our purposes, it is crucial that the functions $A_k$ and $B_k$ are analytic, while their exact forms are not important.

Lemma~\ref{lemma:akbkeigenvec} and \thmref{thm:steadystate} generalize to the case of complex $A_k$ and $B_k$, which allows us to decompose the majorana modes $V_k$ in each momentum sector as $V_k = V_{k,u}\oplus V_{k,d}$, i.e. a direct sum of the preserved modes $V_{k,u}$ and the dissipative modes $V_{k,d}$.
In the long-time limit, the correlation matrix $\Gamma_k$ approaches a block-diagonal form according to \thmref{thm:steadystate}.

The first step of the proof is to show that the steady-state correlation matrix $\Gamma_{d,\rs}$ in the dissipative subspace is analytic in the momentum $k$.
The correlation matrix in the steady state satisfies
\begin{align}
    \Gamma_{k,d}^{(\rs)} = P_{k,d} A_k +  P_{k,d} B_k \Gamma_{d,k}^{(\rs)} B_k^\dag P_{k,d},
\end{align}
where $P_{k,d}$ is the projector onto the dissipative subspace and commutes with $A_k$, $B_k$.
One can formally write $\Gamma_{k,d}^{(\rs)}$ as
\begin{align}
    \Gamma_{k,d}^{(\rs)} = \calL_k^{-1}[P_{k,d} A_k], \quad \calL_k[\Gamma] \equiv \Gamma -  P_{k,d} B_k \Gamma B_k^\dag P_{k,d}.
\end{align}
It is important that $B_{k,d} \equiv P_{k,d} B_k$ has a spectral radius less than unity, indicating that $\calL_k$ is an invertible map.
Hence, the analyticity of $\Gamma_{k,d}^{(\rs)}$ hinges on the analyticity of the projector $P_{k,d}$.

The projector onto the dissipative subspace is analytic except for at most removable singularities at a finite number of momenta according to the following Lemma.
Here, we consider the projector $P_{k,u}$ onto the preserved subspace, which is the orthogonal complement of $P_{k,d}$.
\begin{lemma}\label{lemma:projectoranalytic}
The dimension of the preserved space $V_{k,u}$ is a constant in $k$ except for a finite number of momenta.
Furthermore, the hermitian projector $P_{k,u}$ onto $V_{k,u}$ can be decomposed into two hermitian projectors $\tilde{P}_{k,u}$ and $P'_{k,u}$, i.e. $P_{k,u} = \tilde{P}_{k,u} + P'_{k,u}$, with the following properties: (1) $P'_{k,u} = 0$ except for a finite number of momenta; (2) $\tilde{P}_{k,u}$ can be extended to an analytic function of $k$ in the vicinity of the real axis.
\end{lemma}

We leave the proof in Appendix~\ref{app:proofLemProjectorAnalytic}. 
This lemma implies that the preserved fermion modes decompose into continuous bands, the projector onto which is an analytic function $\tilde{P}_{k,u}$, and a finite number of discrete modes at special momenta captured by $P'_{k,u}$.
The analyticity of $\tilde{P}_{k,u}$ implies that the Hilbert space of the preserved modes in the continuous bands are spanned by exponentially localized Wannier orbitals in real space\footnote{
As an example, we consider a channel circuit made of two-site channels that preserve the fermion modes on the left and propagate them to the right.
These preserved modes are described by the strictly local Wannier orbitals.
It may come as a surprise that the preserved modes in this completely local channel dynamics can have a non-trivial GNVW index~\cite{Kitaev:2005hzj,Gross:2012ygy}.
We also note that only the preserved bands together may have a description in terms of the local Wannier orbitals. 
One can envision a channel with two preserved modes of eigenvalues $b(k)=\pm e^{ik/2}$. 
In this case, one cannot separate the two bands such that the projector onto each one is analytic in $k$; therefore, there is no local Wannier orbitals for each band.
}.

Next, we show that the correlation matrix $\Gamma_{k,d}^{(\rs)}$ determines the entanglement spectrum of the 2D state prepared in the corresponding sequential circuit.
Considering the reduced density matrix $\rho$ of the region above a horizontal cut in the 2D system, the many-body entanglement Hamiltonian is
\begin{align}
H_E \equiv -\log(\rho) =\frac{1}{2}\sum_{k} \sum_{\alpha,\beta} h_{\alpha\beta,k} c_{\alpha,-k} c_{\beta,k} + \text{const.}\,, 
\end{align}
where $h_k$ is the single-particle entanglement Hamiltonian, directly related to the correlation matrix $\Gamma_k$ in the subregion
\footnote{The relation between entanglement spectrum and the correlation matrix can be easily derived using the eigenbasis of the entanglement Hamiltonian. If a complex eigenmode $a_\alpha = c_1 + ic_2$ has an ``entanglement energy'' of $\varepsilon_\alpha$, the average occupation number $i\<c_1 c_2\> = e^{-\varepsilon_\alpha}/(1+e^{-\varepsilon_\alpha})$. The corresponding eigenvalue of the correlation matrix $i\lambda_\alpha = \pm (i\<c_1 c_2\> - i\<c_2 c_1^\dag\>)/2$. Therefore $\varepsilon_{\alpha} = \log\frac{1+\lambda_\alpha}{1-\lambda_\alpha}$.
}
\begin{equation}
    h_k = \log \frac{I - i\Gamma_k}{I + i \Gamma_k}\,.
    \label{Eq:entHam}
\end{equation}
The single-particle entanglement Hamiltonian satisfies $h^\dag_k = h_k$ and $h^*_k = -h_{-k}$, and its eigenvalues, i.e. the entanglement spectrum, obey $\varepsilon_k = -\varepsilon_{-k}$.
Thus, only the eigenmodes for $k\in[0,\pi]$ are independent.

In a sequentially prepared state, the entanglement spectrum associated with a horizontal cut is the same as the spectrum of the evolved density matrix in the channel dynamics.
Thus, one can replace $\Gamma_k$ in \eqnref{Eq:entHam} with the correlation matrix in the channel dynamics.
Note that the entanglement spectrum has contributions from the bulk and the boundary as explained at the end of Sec.~\ref{subsec:GaussianChannelKitaevChain}.
The bulk entanglement spectrum, pertaining to the bulk correlation in the sequentially prepared state, is governed by a single-particle entanglement Hamiltonian determined by $\Gamma_{k,d}^{(\rs)}$,
\begin{equation}
    h_\text{bulk,k} \equiv \log \frac{I - i\Gamma_{k,d}^{(\rs)}}{I + i \Gamma_{k,d}^{(\rs)}}\,.
    \label{Eq:bulkentHam}
\end{equation}
Since $\Gamma_{k,d}^{(\rs)}$ is analytic up to removable singularities, the eigenvalues of $\Gamma_{k,d}^{(\rs)}$ and also the bulk entanglement spectrum must be continuous up to removable singularities\footnote{The entanglement spectrum may diverge symmetrically at a certain momentum $k_0$: $\lim_{k\rightarrow k_0^+}\epsilon_k= \lim_{k\rightarrow k_0^-}\epsilon_k = \pm\infty$ when $\Gamma_{k_0,d}^{(\rs)}\rightarrow \pm 1$. Physically, it means that one mode is completely occupied/empty. We regard this type of divergence as a removable singularity.}.
In particular, such a bulk entanglement spectrum is non-chiral, because it cannot change discontinuously from a positive to a negative value.
\begin{thm}
\label{thm:nogo}
The state prepared by a translationally invariant Gaussian fermion sequential circuit has a continuous bulk entanglement spectrum up to removable singularities in the thermodynamic limit. In particular, the bulk entanglement spectrum is non-chiral.
\end{thm}
The proof of this theorem uses only the analyticity of $A_k$ and $B_k$. It applies to the case where the sequential circuit is not strictly local in space, e.g. when each $U_{x,t}$ in Eq.~\eqref{eq:sequential circuit} has exponential tails in the $x$ direction.

\begin{figure}[t]
\centering
\includegraphics[width=0.8\textwidth]{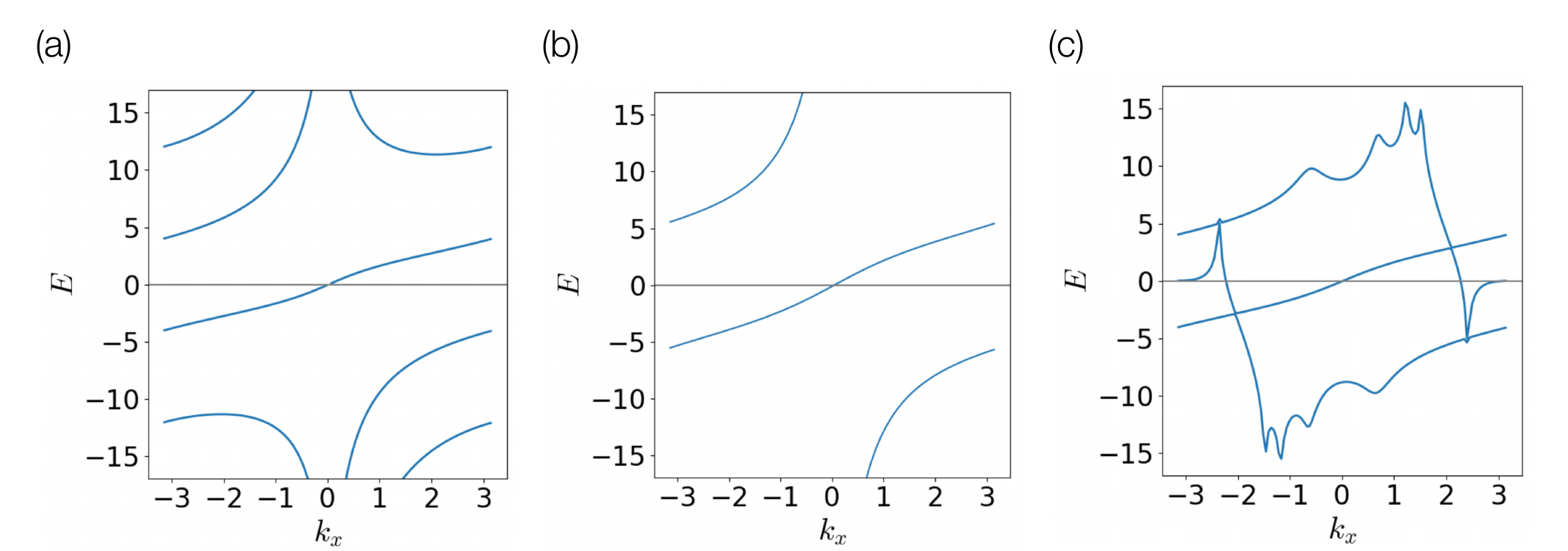}
\hspace{10pt}
\caption{Single-particle entanglement spectrum of the ground state of a p+ip superconductor and its tensor network approximations. (a) Entanglement spectrum of the ground state of the Hamiltonian in Eq.~\eqref{Eq:pipHamiltonian}. (b) Entanglement spectrum of a chiral tensor network state~\cite{Dubail:2013pda}. The entanglement spectrum is discontinuous at $k_x = 0$. (c) Entanglement spectrum of an isoTNS approximation of the ground state. The isoTNS entanglement spectrum mimics the chiral mode. However, the continuity of isoTNS entanglement spectrum requires extra anti-chiral modes, leading to a trivial entanglement spectrum.}
\label{Fig:entanglementspec}
\end{figure}

To gain intuition of this result, we demonstrate the difference between the entanglement spectra in a chiral state and its approximations by general 2D tensor network states and 2D isoTNS in a concrete example.
Fig.~\ref{Fig:entanglementspec}(a) shows the entanglement spectrum of the ground state of a $p+ip$ superconductor with the following mean-field Hamiltonian
\begin{align}
\label{Eq:pipHamiltonian}
    H = \sum_{\mathbf{r}}(-a_{\mathbf{r}}^\dag a_{\mathbf{r} + \hat{x}} - a_{\mathbf{r}}^\dag a_{\mathbf{r} + \hat{y}}
    -a_{\mathbf{r}}a_{\mathbf{r} + \hat{x}} -i  a_{\mathbf{r}}a_{\mathbf{r} + \hat{y}}) + h.c. + \sum_{\mathbf{r}}2a_{\mathbf{r}}^\dag a_{\mathbf{r}}\,,
\end{align}
where $a_{\mathbf{r}},a_{\mathbf{r}}^\dag$ are the annihilation and creation operators of the complex fermion at site $\mathbf{r}$.
The single-particle entanglement spectrum contains $2y$ branches (where $y$ is the location of the cut) and has a chiral mode passing zero at $k = 0$.
When the cut is away from the bottom and top boundaries, energy eigenmodes far from the cut are almost completely occupied or completely empty up to exponential tails, giving rise to an entanglement spectrum linearly proportional to the distance from the cut.
In the thermodynamic limit, one branch changes discontinuously from $+\infty$ to $-\infty$ at $k=0$ (not visible in the figure), corresponding to the occupation of the gapless chiral edge mode at $k=0$ on the bottom and top edges (see Appendix~\ref{app:entanglementspectrum} for details).

Fig.~\ref{Fig:entanglementspec}(b) shows the entanglement spectrum of a chiral tensor network state that represents a $p+ip$ superconductor with algebraically decaying correlation in the bulk~\cite{Dubail:2013pda}. Its entanglement spectrum also has a chiral mode passing $0$ at $k=0$.
However, the entanglement spectrum of a tensor network state always has a finite number of branches, limited by the number of virtual modes on each bond. 
In this case, having a chiral entanglement spectrum requires a discontinuity in the spectrum. 
Note that the entanglement spectrum jumps from $+\infty$ to $-\infty$ at $k = 0$.
In Appendix~\ref{app:entanglementspectrum}, we discuss the relation between the chiral entanglement spectrum and Chern number in the presence of algebraically decaying correlation in the bulk.

Fig.~\ref{Fig:entanglementspec}(c) shows the bulk entanglement spectrum of the isoTNS optimized to approximate the ground state of the Hamiltonian in Eq.~\eqref{Eq:pipHamiltonian} on a $36 \times 36$ lattice.
The spectrum reproduces the middle branch of Fig.~\ref{Fig:entanglementspec}(a), but has extra anti-chiral modes, as required by the continuity of the entanglement spectrum.
This gives rise to a non-chiral entanglement spectrum,
We note that the entanglement spectrum of isoTNS may have cusps in the form of $(k-k_0)^{1/n},\, n\in\mathbb{Z}$ as shown in Fig.~\ref{Fig:entanglementspec}(c).
Here, for the convenience of numerical simulation, the isoTNS we used is rotated by 45 degrees compared to Fig.~\ref{fig:correspondence channel isotns}(b).
Our results still apply to this case, as we explain in Sec.~\ref{sec:othercircuits}.

To close this section, we discuss the implications of Lemma~\ref{lemma:projectoranalytic} on correlation functions in the steady state of channel dynamics and in the 2D state.
First, the correlation matrix $\Gamma_{k,d}^{(\rs)}$ being analytic with at most removable singularities implies the following corollary.
\begin{cor}
In a translationally invariant 1D local Gaussian fermion channel, the correlation functions between dissipative modes decay exponentially in the steady state in the thermodynamic limit.\label{thm:exponentialdecay}
\end{cor}
\begin{proof}
By Lemma~\ref{lemma:projectoranalytic}, the projector onto dissipative modes $P_{k,d}$ is analytic in the vicinity of the real axis except for at most a finite number of momenta where $P_{k,u}\neq\tilde{P}_{k,u}$.
Since $\Gamma_{k,d}^{(\rs)} = \mathcal{L}_k^{-1}[A_k]$ represents a correlation matrix for real $k$, its matrix elements are bounded.
This implies that $\Gamma_{k,d}^{(\rs)}$ is analytic up to at most removable singularities.
Fourier transforming $\Gamma_{k,d}^{(\rs)}$ to real space, the removable singularity gives $\calO(1/L)$ contribution, which is negligible in the thermodynamic limit, and the analytic part leads to an exponentially decaying correlation function.\footnote{When using numerically optimized isoTNS to approximate a gapped chiral state in a finite-size system, the fermion two-point correlation can asymptote to an $\calO(1/L)$ value~\cite{wu2025alternatinggaussianfermionicisometric}.} 
\end{proof}

We note that dissipative modes are defined in the momentum space and are in general not local in real space. 
We further note that for any 1D local translation-invariant Gaussian fermion channel, one can choose an initial state such that correlation functions of all local operators in the steady state decay exponentially. 
In particular, taking the maximally mixed state as the initial state, the correlation function in the steady state $\Gamma^{(\rs)}_k = (0,0;0,\Gamma_{k,d}^{(\rs)})$ is analytic in $k$, therefore, decays exponentially in space.

We note that, if the channel dynamics indeed generated a 2D chiral state, the correlation matrix $\Gamma_{k,d}^{(\rs)}$ would have a chiral spectrum, i.e. with discontinuities in the momentum space.
This would indicate the correlation function in the real space to decay algebraically in the form $1/x$ with $x$ being the separation.

Second, deep in the bulk of the 2D isoTNS, the correlation functions in physical state between operators on the same row are determined by the correaltion between dissipative modes in the steady state, giving rise to the following corollary.
\begin{cor}
\label{cor:exponentialdecayphysical}
    The physical correlation function of a 2D Gaussian fermion isoTNS decays exponentially along the spatial direction.
\end{cor}
In contrast, a 2D Gaussian fermion isoTNS may have algebraically decaying correlation in the temporal direction.
When the corresponding channel has a discrete preserved mode at momentum $k_0$, the eigenvalue $b_k$ has a norm arbitrarily close to unity as $k\rightarrow k_0$; by \thmref{thm:correlationlength}, the correlation length in the time direction is infinite. 
This asymmetry between the temporal direction and spatial direction is also observed for interacting isoTNS~\cite{wu2025alternatinggaussianfermionicisometric,Liu:2023gvd}. 
\secref{sec:othercircuits} discusses the correlation functions along other directions in the isoTNS and in other types of isometric tensor network states.

\subsection{Generalization to other sequential circuit architectures}
\label{sec:othercircuits}

We have established the no-go theorem for a specific sequential circuit architecture (\thmref{thm:nogo}). 
Our result implies that the correlation functions decay exponentially in the prepared 2D state (isoTNS) along the $x$-direction (Cor.~\ref{cor:exponentialdecayphysical}, illustrated in Fig.~\ref{Fig:isoTNSspacelikealtisoTNS}(a)).
In this section, we first consider other sequential circuit architectures that sequentially generate the same isoTNS, but along different directions.
We show that the correlation function decays exponentially along other directions and the entanglement spectrum along different cuts in the same isoTNS are non-chiral (Fig.~\ref{Fig:isoTNSspacelikealtisoTNS}(b-c)). 
We then generalize our results to a sequential circuit that generates an isoTNS with a different isometric structure (Fig.~\ref{Fig:isoTNSspacelikealtisoTNS}(d)).
To generate such an isoTNS over $L\times L$ sites, the sequential circuit is of depth $\calO(L^2)$.

\begin{figure}[t]
\centering
\includegraphics[width=0.7\textwidth]{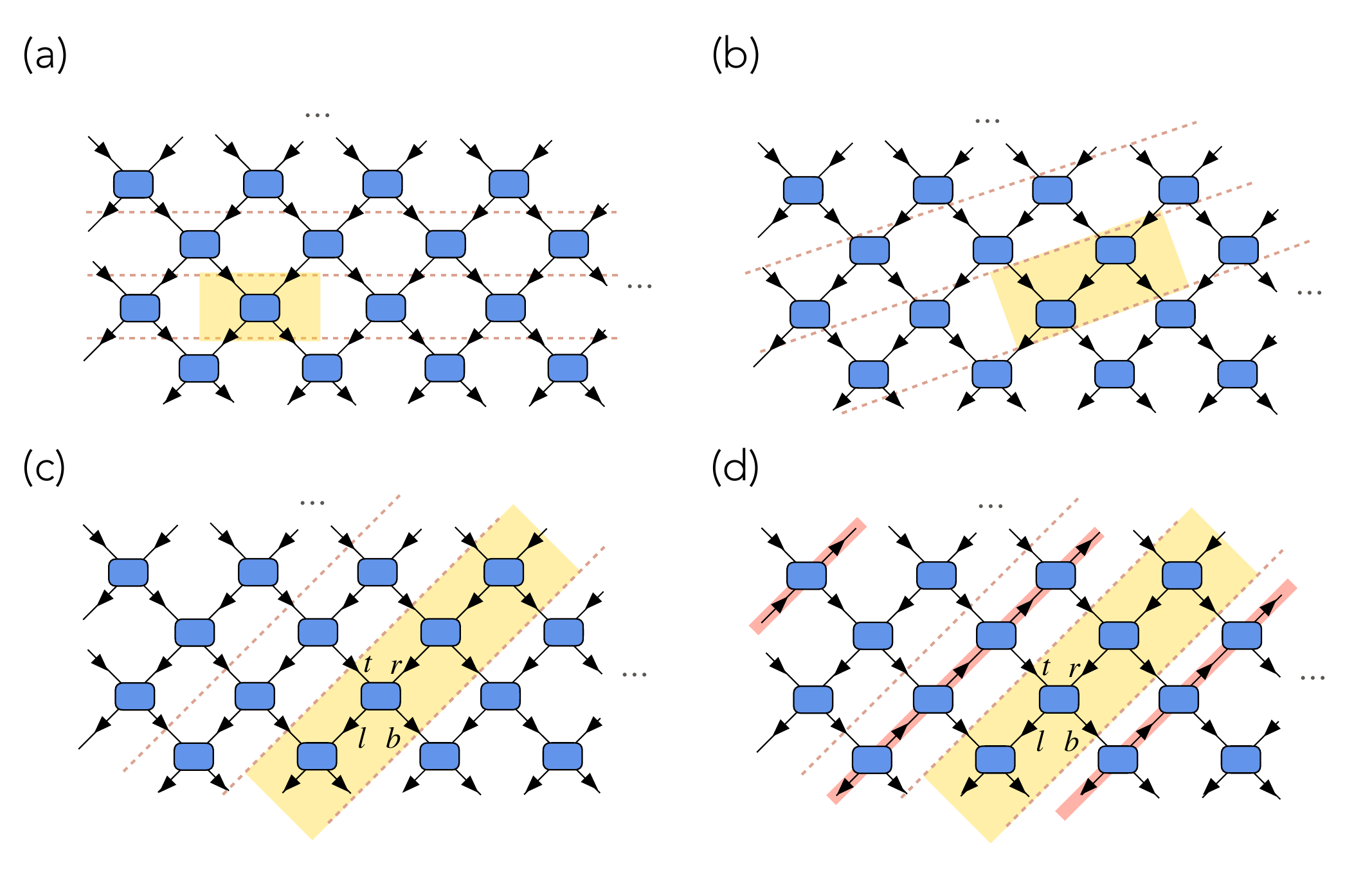}
\hspace{10pt}
\caption{(a-c) virtual states along the space direction, a space-like direction, and the light direction in isoTNS. For any space-like direction, the evolution of the virtual state from one cut to the cut above is given by a finite-depth local channel, where elementary channels in different unit cells (yellow-shaded area) can be applied in parallel. For the light-like direction, the evolution of a virtual state to the cut above is given by a sequential channel. In panel (c), we use $l,r,t,b$ to label the left, right, top, and bottom virtual legs of each tensor. The yellow shaded channels act sequentially from the left to the right, and they compose into a channel from all `$b$' legs along the cut to all `$t$' legs along the cut above. (d) alt-isoTNS: isometric arrows along the red lines are reversed relative to panel (c). The resulting alternating-isoTNS corresponds to a depth-$L^2$ sequential circuit. The evolution of the virtual state from each cut to the cut above is given by a sequential channel, but the order in which local channels act alternates between even and odd steps.}
\label{Fig:isoTNSspacelikealtisoTNS}
\end{figure}

To begin, we consider sequentially generating the isoTNS in Fig.~\ref{fig:correspondence channel isotns}(c) along a time-like direction (as shown in Fig.~\ref{Fig:isoTNSspacelikealtisoTNS}(b)).
The evolution of the virtual state on one space-like cut for one step can be described by a local channel. 
In particular, if the cut has a rational slope, the local channel is translationally invariant (with a unit cell marked by the yellow-shaded area). 
The results in Sec.~\ref{subsec:nogo_sequential} still apply to these circuits.
In particular, the bulk entanglement spectrums along these spacelike cuts are non-chiral (\thmref{thm:nogo}), and the physical correlation function of the isoTNS decays exponentially along these directions (Cor.~\ref{cor:exponentialdecayphysical}).

The sequential generation of the isoTNS along the light-like direction in Fig.~\ref{Fig:isoTNSspacelikealtisoTNS}(c) is characterized by a sequential channel, where the local channels in the yellow-shaded region must be applied sequentially from the left to the right. 
Nonetheless, this Gaussian fermion channel is translation-invariant, therefore, described by $\mathcal{N}[\Gamma_k] = A_k + B_k\Gamma_kB_k^\dag$, where $\Gamma_k$ is the correlation matrix of the virtual state along a cut in the light-like direction.
Crucially, $A_k$ and $B_k$ are given by rational trigonometric polynomials of $k$ as shown in Appendix~\ref{app:GfTNS}.
Since $A_k$ and $B_k$ are bounded along the real axis, they are analytic in the vicinity of the real axis except at removable singularities.
\thmref{thm:nogo} and Cor.~\ref{cor:exponentialdecayphysical} still apply to this case, as the proofs only rely on the analyticity of $A_k$ and $B_k$.

Our results in Sec.~\ref{subsec:nogo_sequential} also apply to the alt-isoTNS~\cite{wu2025alternatinggaussianfermionicisometric}, which has an alternative isometric structure (shown in Fig.~\ref{Fig:isoTNSspacelikealtisoTNS}(d)) with the arrows pointing to alternating directions along the $(1,1)$ direction.
The alt-isoTNS can be generated by sequentially applying local unitary gates in a snake order (following the arrows in the reverse direction), which can be arranged in a depth-$L^2$ circuit.
In this case, there is still a channel-state correspondence; one can regard the $(1,-1)$ direction as the temporal direction, and the $(1,1)$ direction as the spatial direction.
Each time step involves $L$ local channels in a sequential order and evolves the virtual state by a sequential channel $\mathcal{N}[\Gamma_k] = A_k + B_k\Gamma_kB_k^\dag$, with different $A_k$ and $B_k$ for even and odd time steps. 
Nonetheless, $A_k$ and $B_k$ are analytic, and we can combine two time steps as one period and apply our no-go result to this case.
Thus, we conclude that the physical correlation function along the space direction of alt-isoTNS decays exponentially, and alt-isoTNS cannot represent chiral states in the thermodynamic limit.

\section{No-go theorem for interacting sequential circuits}
\label{sec:interacting}

In generic interacting systems, tripartite entanglement measures probe the chiral edge mode and, more generally, un-gappable edges in topological states~\cite{Zou:2020bly,Siva:2021cgo}.
In this section, we prove a no-go theorem by showing that the state prepared by sequential circuits cannot host the tripartite entanglement of a chiral state due to the constraint of causality.
Our theorem makes mild assumptions on the circuit geometry and initial state and does not rely on spacetime translational invariance.

The chiral state---or, more generally, any state with an ungappable edge--- supports gapless degrees of freedom localized on the upper and lower edges, which leave signatures in the tripartite entanglement measures.
To probe this entanglement structure, we consider an $L_x \times L_y$ cylinder with periodic boundary conditions in the $x$ direction and open boundaries in $y$ and partition the cylinder into three rectangular regions $A$, $B$, and $C$, each of size $L_x/3\times L_y$ (illustrated in Fig.~\ref{fig:triangle state}). 
As a key entanglement feature, the \emph{Markov gap} $h(A:B)$ takes the value $\frac{c}{3}\log 2$, which depends on the central charge $c$ of conformal field theory (CFT) describing the un-gappable edge of the state.
This implies that the wave function $\psi_{ABC}$ of a chiral state cannot be a ``sum of triangle states'' defined below~\cite{Zou:2020bly,Siva:2021cgo}.
\begin{defn}[Triangle state]
    A pure state $\psi_{ABC}$ is a triangle state if there exists decomposition $\mathcal{H}_A =  \mathcal{H}_{A_0}\oplus(\mathcal{H}_{A_L}\otimes \mathcal{H}_{A_R})$, $\mathcal{H}_B = \mathcal{H}_{B_0}\oplus(\mathcal{H}_{B_L}\otimes \mathcal{H}_{B_R})$, and $\mathcal{H}_C = \mathcal{H}_{C_0}\oplus(\mathcal{H}_{C_L}\otimes \mathcal{H}_{C_R})$ such that
    \begin{align}
    \label{eq:triangle state}
        \psi_{ABC} = \psi_{A_R B_L}\otimes \psi_{B_R C_L}\otimes \psi_{C_R A_L},
    \end{align}
    where $\psi_{A_R B_L}\in \mathcal{H}_{A_R}\otimes\mathcal{H}_{B_L}$, $\psi_{B_R A_L}\in \mathcal{H}_{B_R}\otimes\mathcal{H}_{A_L}$, and $\psi_{C_R A_L}\in \mathcal{H}_{C_R}\otimes\mathcal{H}_{A_L}$.
\end{defn}

\begin{defn}[Sum of triangle states]
    A pure state $\psi_{ABC}$ is a sum of triangle state if there exists decomposition
    $\mathcal{H}_A = \mathcal{H}_{A_0}\oplus(\oplus_j\mathcal{H}_{A^j_L}\otimes \mathcal{H}_{A^j_R})$, $\mathcal{H}_B = \mathcal{H}_{B_0}\oplus(\oplus_j\mathcal{H}_{B^j_L}\otimes \mathcal{H}_{B^j_R})$, and $\mathcal{H}_C = \mathcal{H}_{C_0}\oplus(\oplus_j\mathcal{H}_{C^j_L}\otimes \mathcal{H}_{C^j_R})$ such that
    \begin{align}
    \label{eq:sum of triangle states}
        \psi_{ABC} = \sum_j \sqrt{p_j}\psi_{A^j_R B^j_L}\otimes \psi_{B^j_R C^j_L}\otimes \psi_{C^j_R A^j_L},
    \end{align}
    where $\sum_j p_j = 1$, $\psi_{A^j_R B^j_L}\in \mathcal{H}_{A^j_R}\otimes\mathcal{H}_{B^j_L}$, $\psi_{B^j_R A^j_L}\in \mathcal{H}_{B^j_R}\otimes\mathcal{H}_{A^j_L}$, and $\psi_{C^j_R A^j_L}\in \mathcal{H}_{C^j_R}\otimes\mathcal{H}_{A^j_L}$.
\end{defn}

We now consider a 2D sequential circuit on the $L_x\times L_y$ cylinder (as in Fig.~\ref{fig:correspondence channel isotns}(c)). 
When the circuit depth satisfies $L_y<L_x/6$, the causality of the circuit leads to the following no-go theorem.
\begin{thm}
Let $\psi_{ABC}$ on the $L_x\times L_y$ cylinder be a state prepared by a sequential circuit of depth $L_y<L_x/6$.
\begin{enumerate}
\item If the initial state is a trivial product state, then $\psi_{ABC}$ is a triangle state.
\item If the initial state on the bottom edge is a cat state $\sum_{i}|ii\dots i\rangle$ (with the bulk initialized in a product state), then $\psi_{ABC}$ is a sum of triangle states.
\end{enumerate}
In either case, the output state cannot have an ungappable edge.
    \label{thm:gabhab}
\end{thm}
\begin{proof}
It is clearest to state the proof using the isoTNS representation of the sequential circuit.
Since $L_y < L_x/6$, we can cut the sequential circuit into six regions along the light cone as shown in \figref{fig:triangle state}.
We define the concatenation of all unitary gates in the region $B$ and above the light cone by $U_B$, similarly $U_A$ and $U_C$ for regions $A$ and $C$.
And we examine the state before and after applying the gates above the light cone. 
Let $\calH_{B_L}$ ($\calH_{B_R})$ denote the Hilbert space that comprises physical legs below the left (right) light cone in $B$ and virtual legs cut by the left (right) light cone, similarly for region $A$ and $C$.
It is easy to see that the state prepared by unitaries below the light cones is a triangle state $\psi_{A_R,B_L}\otimes\psi_{B_R,C_L}\otimes\psi_{C_R,A_L}$, where the subscripts denotes the associated Hilbert space. 
The final state prepared by the sequential circuit is
\begin{align}
    \psi_{ABC} = U_A\otimes U_B\otimes U_C[\psi_{A_R,B_L}\otimes\psi_{B_R,C_L}\otimes\psi_{C_R,A_L}\otimes |0\dots0\rangle],
\end{align}
where $|0\dots0\rangle$ represents the input states of $U_A$, $U_B$ and $U_C$ that is not along the cut.
This final state is related to the triangle state by changing the basis in $A$, $B$, and $C$ separately. Thus, it is still a triangle state.

\begin{figure}[t]
\centering
\includegraphics[width=0.85\textwidth]{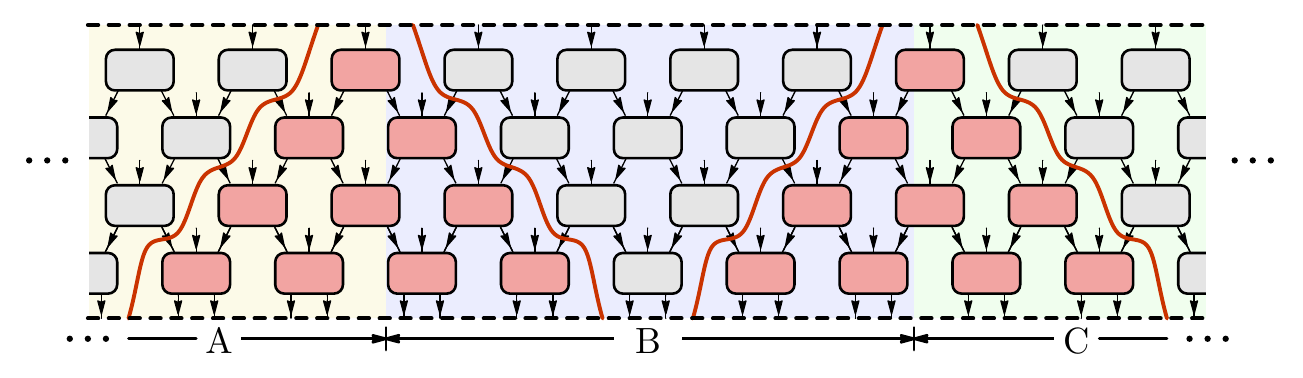}
\caption{Sequential circuit on a ring, $L_y< L_x/6$ with the periodic boundary condition in the $x$ direction. The cylinder is divided into three blocks of size $L_x/3\times L_y$; the left half of $A$ and the right half of $C$ are omitted in the figure. The grey-colored gates above the red cuts form three clusters $U_A$, $U_B$, and $U_C$. The pink-colored gates below the red cuts prepare a triangle state from a product state or a sum of triangle states from a cat state.}
\label{fig:triangle state}
\end{figure}

If the input state is a cat state at the bottom edge, the final state reads
\begin{align}
    \psi_{ABC} = U_A\otimes U_B\otimes U_C[\sum_i\psi_{A_R,B_L}(i)\otimes\psi_{B_R,C_L}(i)\otimes\psi_{C_R,A_L}(i)\otimes|ii\dots i\rangle_{ABC}\otimes |0\dots0\rangle],
    \label{eq:cylindersumoftriangle}
\end{align}
where $|ii\dots i\rangle_{ABC} = |ii\dots i\rangle_{A}\otimes |ii\dots i\rangle_{B}\otimes|ii\dots i\rangle_{C}$ represents the input state of $U_A$, $U_B$, and $U_C$ on the bottom edge.
Identify $|ii\dots i\rangle_A\otimes \mathcal{H}_{A_L}\otimes\mathcal{H}_{A_R}$ as $ \mathcal{H}_{A^i_L}\otimes\mathcal{H}_{A^i_R}$, and similarly for $B$ and $C$, the state before the action of $U_A$, $U_B$, and $U_C$ is a sum of triangle states. The action of $U_A\otimes U_B\otimes U_C$ preserves this property.
\end{proof}

The no-go theorem can be further extended to a broad class of initial states, including those obtained by finite-depth local unitary circuits acting on a cat state.
The proof is straightforward as one can consider the finite-depth local unitary that prepares the initial state as a part of the sequential circuit.
We note that such initial states in the thermodynamic limit can approximate any matrix product state (MPS) of finite bond dimension.
Even if the initial state is in the form of MPS, we still expect that the state prepared by a sequential circuit has a vanishing Markov gap $h(A:B)$ and is therefore non-chiral.
This is because $h(A:B)$ is continuous under small changes in the wave function~\cite{Dutta:2019gen,Zou:2020bly}. Taken together these arguments lead to the no-go theorem 2 stated in \secref{sec:intro}.

Note that the 1D brick wall circuit is a special case of the 2D sequential circuit, leading to the following corollary.
\begin{cor}
\label{cor:nogo1DCFT}
    A 1D brick wall circuit on a ring of length $L$, with a depth less than $L/6$, cannot prepare the ground state of a conformal field theory from a trivial product state.
\end{cor}
Corollary~\ref{cor:nogo1DCFT} is consistent with the observation that the quench rate must scale as $\sim 1/L$ in order to reach a CFT ground state from a trivial product state~\cite{Zurek:2005kod}.

\section{Discussion}
\label{sec:discussion}

In this work, we examine the capability of 2D sequential circuits by leveraging the correspondence between the sequential circuit, 1D quantum channel dynamics, and isoTNS.
For Gaussian fermion systems, we prove a no-go theorem for the preparation of chiral topological states with a local 2D translationally invariant sequential circuit.
For generic interacting systems, we prove that sequential circuits with certain constraints cannot generate non-trivial tripartite entanglement, which is required to support chiral edge modes.

We remark that Ref.~\cite{Chen:2024fvu} does provide explicit protocols to prepare chiral states using sequential Gaussian fermion evolution.
Their results do not contradict ours due to the following subtle but crucial differences.
Their first protocol works with degrees of freedom that are continuous in one spatial direction and, therefore, cannot be formulated in tensor networks defined in discrete space.
Their second protocol prepares a chiral state on the lattice using a sequence of adiabatic evolution generated by a gapped Hamiltonian.
The evolution in each step can be recast as a \emph{finite-time} quasi-adiabatic evolution generated by a quasi-local Hamiltonian that has sub-exponential tails in both directions~\cite{Hastings:2010vzr}.
Such an evolution is not captured by our sequential circuits with strict locality in the sequential direction; it maps to 1D open system dynamics that is not Markovian.
Our no-go theorem implies that this protocol fails if one truncates the tail in the sequential direction and simultaneously only allows exponential tails in the other direction.
Identifying the minimal necessary non-locality is left as an open question.

It is interesting to further explore the possibility of preparing chiral states using generic interacting sequential circuits.
Our reasoning suggests that the candidate sequential circuit must correspond to 1D channel dynamics that either have a critical relaxation or have multiple steady states.
In the former scenario, it is likely that the prepared 2D state has an algebraically decaying correlation function~\cite{Dubail:2013pda}.
Our no-go results concern the long-range behavior of correlation functions of a TNS with a finite bond dimension. 
From the perspective of numerical studies of TNS, it would be interesting to investigate how bond dimensions of a TNS scale with the system size in order to achieve a desired accuracy for a chiral topological phase.

Our proof of the no-go results hinges on the channel-state correspondence, which deserves further study in its own right.
For example, a previous work utilizes the correspondence between isoTNS and channel dynamics to investigate the complexity of tensor network algorithms~\cite{Malz:2024val}. 
It is interesting to invoke critical quantum/classical channel dynamics to construct states with high complexity~\cite{Gopalakrishnan:2023kdp}.
It is also interesting to resort to the entanglement structure of two-dimensional TNS to constrain the capabilities of the channel dynamics. 
In particular, the presence of preserved information in the channel dynamics seems to be of special interest across this correspondence. 
On the state side, it is likely related to having non-trivial topological order.
On the dynamics side, it is naturally related to the question of passive/self-correcting memory. It is interesting to leverage this correspondence to explore self-correcting classical memory in one dimension and quantum memory in two dimensions.

\section*{Acknowledgment}
ZD acknowledges helpful discussions with Michael Zaletel and Xiao-Liang Qi.
RF and ZD acknowledge the hospitality of Kavli Institute for Theoretical Physics during their visit.
YW is supported by a start-up grant from IOP-CAS.
R.F. is supported by the Gordon and Betty Moore Foundation (Grant GBMF8688).
Y.B. is supported by grant GBMF7392 from the Gordon and Betty Moore Foundation to University of California, Santa Barbara Kavli Institute for Theoretical Physics (KITP).
This research was supported in part by grant NSF PHY-2309135 to the University of California, Santa Barbara Kavli Institute for Theoretical Physics (KITP).

\bigskip 

\emph{Note added.---}During the completion of the manuscript, we became aware of an independent
work using the correspondence to address related questions~\cite{toappear}. We thank them for coordinating the submission.

\appendix

\section{Topology of Gaussian fermion matrix product states}
\label{app:KitaevChain}
The 0D Gaussian channel dynamics that preserve an even and an odd number of Majorana modes exhibit a physical distinction, corresponding to 1D MPS in the trivial and topological phase, respectively.

As a key observation, an even number of preserved modes can form a many-body Hilbert space, undergoing a unitary evolution, while an odd number of preserved Majorana modes must combine with an additional Majorana mode $\gamma_0$ to form an integer-dimensional subspace. 
Since $\gamma_0$ is \textit{not} a preserved mode, the correlation functions between $\gamma_0$ and other modes will decay exponentially. 
Repeated actions of the channel give a fermionic version of the dephasing channel, leading to the following corollary of \thmref{thm:steadystate}.
\begin{cor}
\label{cor:evenoddmajorana}
In the case that the preserved space $V_u$ contains an even number of modes, the output state of the channel dynamics in the long-time limit is given by
\begin{align}
    \lim_{t\rightarrow\infty} \mathcal{N}^t[\rho] = \mathcal{U}^t\Tr_{\mathcal{H}_d}(\rho)(\mathcal{U}^\dag)^t\otimes\rho^{(\rs)}_d\,,
\end{align}
where $\mathcal{H}_u$ ($\mathcal{H}_d$) denotes the many-body Hilbert space corresponding to single-particle fermion modes in $V_u$ ($V_d$), $\Tr_{\mathcal{H}_d}(\cdot)$ denotes the partial trace over the subsystem containing single-particle modes in $V_d$, and $\mathcal{U}$ is the unitary evolution on $\mathcal{H}_u$ induced by the single-particle unitary $U$ in \thmref{thm:steadystate}.
Here, $\rho$ is the many-body density matrix of the initial state, $\rho^{(\rs)}_d$ is the many-body density matrix of the steady state of the dissipative modes.

In the case that $V_u$ contains an odd number of modes, there exists an Majorana mode $\gamma_0$ in $V_d$ such that
\begin{align}
\label{eqapp:oddmajoranasteadystate}
    \lim_{t\rightarrow\infty} \mathcal{N}^t[\rho] = \mathcal{U}^t\frac{\Tr_{\mathcal{H}_d}(\rho) + \gamma_0\Tr_{\mathcal{H}_d}(\rho)\gamma_0}{2}(\mathcal{U}^\dag)^t\otimes\rho^{(\rs)}_d\,,
\end{align}
where $\mathcal{H}_u$ ($\mathcal{H}_d$) be the many-body Hilbert space corresponding to filling Majorana modes in $V_u$ plus $\gamma_0$ (Majorana modes in $V_d$ excluding $\gamma_0$), and $[\mathcal{U},\gamma_0] = 0$.
  
\end{cor}

\begin{proof}
In the case that $V_u$ has an even dimension, $V_d$ also has an even dimension as the total number of Majorana modes is even.
For the preserved modes, the single-particle unitary evolution $U^t$ in \thmref{thm:steadystate} induces a many-body unitary $\mathcal{U}^t$ in $\calH_u$. 
For the dissipative modes, the steady state correlation matrix determines the many-body density matrix $\rho_d^{(\rs)}$ in the subspace $\mathcal{H}_d$.

In the case that $V_d$ has an odd dimension, the real-antisymmetric correlation matrix $\Gamma^{(\rs)}_d$ has at least one zero eigenvector, corresponding to a Majorana mode $\gamma_0$. 
Thus, we can define a many-body Hilbert space $\mathcal{H}_d$ corresponding to filling the even number of Majorana modes in $V_d$ excluding $\gamma_0$; $\Gamma^{(\rs)}_d$ restricting to this space gives a many-body state $\rho^{(\rs)}_d$ in $\mathcal{H}_d$.
We define the complementary many-body Hilbert space as $\mathcal{H}_u$, which is associated with filling $\gamma_0$ and Majorana modes in $V_u$.
The unitary evolution $\mathcal{U}$ in the many-body space $\mathcal{H}_u$ stems from the unitary transformation of these single-particle modes, $(U, 0; 0, 1)$, where $U$ is the unitary acting on $V_u$ as in \thmref{thm:steadystate}.
We note that the correlation function between $\gamma_0$ and any other Majorana mode approaches zero under repeated action of the channel.
This effect can be captured by the fermionic dephasing channel $\rho\rightarrow (\rho+\gamma_0\rho\gamma_0)/2$ for any parity-even density matrix $\rho$.
\end{proof}

\begin{figure}
\centering
\includegraphics[width=0.45\textwidth]{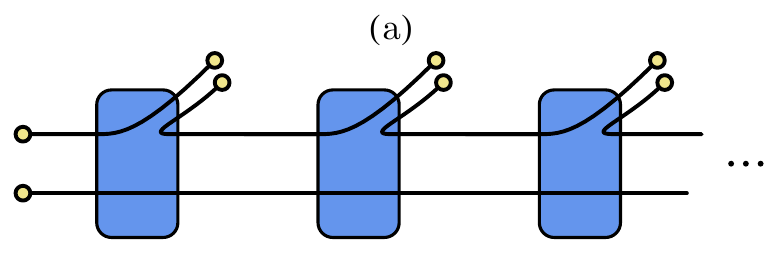}
\hspace{10pt}
\includegraphics[width=0.15\textwidth]{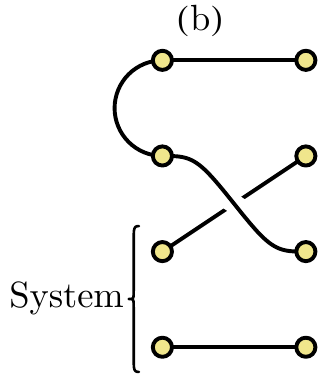}
\caption{(a) Matrix product state representation of the Kitaev Chain. Each yellow dot represents a Majorana mode and each block represents a Gaussian fermion tensor. This tensor can be interpreted as a Gaussian fermion channel that preserves one Majorana mode. Time direction goes from the left to the right. (b) The unitary circuit that realizes the quantum channel in (a).}
\label{Fig:Kitaevchain}
\end{figure}

Under the channel-state correspondence, the two cases, with an even and an odd number of preserved Majorana modes, map to the trivial phase and the topological phase (i.e. the Kitaev chain), respectively.
The preserved Majorana modes become the Majorana modes on the physical boundary.
The simplest example is illustrated in Fig.~\ref{Fig:Kitaevchain}, where each virtual leg contains two Majorana modes; one Majorana mode is preserved, while the other Majorana mode is dissipated into the environment.
This channel can be realized by the unitary gate (Fig.~\ref{Fig:Kitaevchain}(b)) that swaps a Majorana fermion in the system with a Majorana fermion in the environment.
It prepares a fixed-point Kitaev chain in the environment.

The fact that the many-body entanglement spectrum of a Kitaev chain is the same in the sectors with even and odd parity~\cite{Mortier:2024qtg} follows directly from \thmref{thm:steadystate}.
By the channel-state correspondence, the entanglement spectrum is identical to the spectrum of the steady-state density matrix given by Eq.~\eqref{eqapp:oddmajoranasteadystate}.
Since the map $\rho\rightarrow\gamma_0\rho\gamma_0$ swaps the density matrix in the parity-even and the parity-odd sector and the right hand side of Eq.~\eqref{eqapp:oddmajoranasteadystate} is invariant under this map, the spectrum must be the same in the two sectors.
 
Under Jordan-Wigner transformation, the channel $\rho\rightarrow (\rho + \gamma_0\rho\gamma_0)/2$ in $\mathcal{H}_u$ (Eq.~\eqref{eqapp:oddmajoranasteadystate}) maps to a dephasing channel $\rho\rightarrow (\rho + \sigma_x\rho\sigma_x)/2$. 
This bosonic dephasing channel preserves one bit of classical information, the eigenvalue of $\sigma_x$, and $\log(\text{dim}(\mathcal{H}_u)) -1$ bit of quantum information.
This preserved classical information is associated with the $\mathbb{Z}_2$ symmetry of the tensor, and it is known to be related to cat states in 1D and topological states in 2D~\cite{Schuch:2010hfh,Lootens:2020mso}.

\section{Gaussian fermion isometric tensor network states}
\label{app:GfTNS}

In the main text, we interpret each tensor in a Gaussian fermion MPS (or isoTNS in 2D) as an isometric map from the bottom virtual leg to the top virtual and physical legs (as in Fig~\ref{Fig:GFMPS}).
In this appendix, we provide an alternative yet equivalent interpretation~\cite{wu2025alternatinggaussianfermionicisometric,PhysRevA.81.052338,PhysRevB.107.125128}, in which each tensor is regarded as a pure state in the Hilbert space defined by the Majorana modes on its physical and all virtual legs, and tensor contraction over virtual legs is imposed by projecting the Majorana modes connected by a virtual leg to a reference state.
We use this formalism to derive the channel dynamics corresponding to the 2D isoTNS along the direction in Fig.~\ref{Fig:isoTNSspacelikealtisoTNS}(c).

To begin, we consider the Gaussian fermion MPS in Fig.~\ref{Fig:GFMPS}.
Each tensor at location $y$ can be represented as a pure state $\ket{\varphi_y}$ over $n_P + 2n_V$ Majorana modes, where $n_P$ and $n_V = 2n$ are the number of Majorana modes on the physical and the top/bottom virtual legs, respectively.
One can specify this pure state by an anti-symmetric correlation matrix $\Lambda$ in the form
\begin{align}
    \Lambda = \begin{pmatrix} \Lambda_{P} & \Lambda_{PV_t} & \Lambda_{PV_b} \\ -\Lambda_{PV_t}^T & \Lambda_{V_t} & \Lambda_{V_tV_b}\\
    -\Lambda_{PV_b}^T & -\Lambda_{V_tV_b}^T & \Lambda_{V_b}
    \end{pmatrix} = \begin{pmatrix} A_{P} & A_{PV} & B_{P} \\ -A_{PV}^T & A & B\\
    -B_{P}^T & -B^T & 0
    \end{pmatrix}\,.
\end{align}
where the subscripts $P, V_t, V_b$ for $\Lambda$ labels the correlation of Majoranas modes on the physical, the top virtual, and the bottom virtual legs, respectively.
Here, $\Lambda^T \Lambda = I$ for pure states, and the isometric condition requires $\Lambda_{V_b} = 0$, i.e. the reduced density matrix in the Hilbert space of the bottom leg is maximally mixed\cite{wu2025alternatinggaussianfermionicisometric}.
In what follows, we show that this correlation matrix $\Lambda$ is exactly what is identified in Eq.~\eqref{eq:tensorcorrelationmatrix} and fully determines the channel $\calN_{V \mapsto VP}$ associated with this isometric tensor.

In this new interpretation of Gaussian fermion MPS, we obtain the quantum state $\ket{\Phi_P}$ over physical Majorana modes by projecting the virtual Majoranas onto a definite state, 
\begin{align}
    \ket{\Phi_P} = \bra{\psi_0}\prod_b \bra{\psi_b^{\text{ref}}} \prod_{y}\ket{\varphi_y},
\end{align}
where the Majorana modes on the bottom row are projected onto the initial state $\ket{\psi_0}$ of the channel dynamics, and Majorana modes on the virtual legs connected by a bond $b$ are projected onto a reference state $\ket{\psi_b^{\text{ref}}}$.
The reference state has definite even parity $ic_{1,j}c_{2,j} = 0$ for $j = 1, 2,\cdots, n_V$, where $c_{1(2),j}$ are Majorana modes on the two virtual legs connected by the bond.

Within this formulation, we can reproduce the output state of the channel dynamics at time step $y$ by contracting all the tensors below the cut at location $y$.
Let $\Gamma_y$ be the correlation matrix on the virtual bonds at step y.
By contracting with the tensors in the next row, we obtain the evolved correlation matrix on the physical and the virtual legs\cite{Bravyifermionlinearoptics},
\begin{align}
    \Gamma_y \mapsto \begin{pmatrix} \Lambda_{P} & \Lambda_{PV_t} \\ -\Lambda_{PV_t}^T & \Lambda_{V_t}
    \end{pmatrix} + \begin{pmatrix}  \Lambda_{PV_b} \\ \Lambda_{V_tV_b}
    \end{pmatrix} \Gamma_y \begin{pmatrix} 
    \Lambda_{PV_b}^T & \Lambda_{V_tV_b}^T
    \end{pmatrix},
\end{align}
which reproduces the channel dynamics, and the correlation matrix $\Lambda$ is given by that in Eq.~\eqref{eq:tensorcorrelationmatrix}.

We now work with this formalism to derive the dynamics generated by the 2D isoTNS along the light-like direction as in Fig.~\ref{Fig:isoTNSspacelikealtisoTNS}(c).
The state $\Gamma_k$ in the channel dynamics lives on the virtual legs along the cut, and its evolution is generated by contracting the state with tensors between two steps, which takes a general form $\mathcal{N}[\Gamma_k] = A_k + B_k\Gamma_kB_k^\dag$. 
Our goal here is to determine $A_k$ and $B_k$.

First, we denote the correlation matrix associated with each individual tensor as
\begin{align}
    \Lambda = \begin{pmatrix} \Lambda_{P} & \Lambda_{PV_r} & \Lambda_{PV_t}& \Lambda_{PV_l} & \Lambda_{PV_b} \\ 
    -\Lambda_{PV_r}^T & \Lambda_{V_r} & \Lambda_{V_rV_t} & \Lambda_{V_rV_l} & \Lambda_{V_rV_b}\\
    -\Lambda_{PV_t}^T & -\Lambda_{V_rV_t}^T & \Lambda_{V_t} & \Lambda_{V_tV_l} & \Lambda_{V_tV_b}\\
    -\Lambda_{PV_l}^T & -\Lambda_{V_rV_l}^T & -\Lambda_{V_tV_l}^T & 0 & 0\\
    -\Lambda_{PV_b}^T & -\Lambda_{V_rV_b}^T & -\Lambda_{V_tV_b}^T & 0 & 0\\
    \end{pmatrix}
\end{align}
where the subscript $V_{l,r,t,b}$ denotes the virtual legs on the left ($l$), the right ($r$), the top ($t$), and the bottom ($b$), as illustrated in Fig.~\ref{Fig:isoTNSspacelikealtisoTNS}.
The tensor can also be interpreted as an isometric map from the virtual legs $l$, $b$ to the physical leg and the virtual legs $t$ and $r$.

The key step is to determine the correlation matrix of the Majorana modes on the top and the bottom virtual legs, after contracting the left ($l$) and right ($r$) virtual legs of each tensor with their neighbors along the cut.
This correlation matrix in the momentum space can be derived using the Grassmann integral~\cite{Bravyifermionlinearoptics},
\begin{equation}
\begin{aligned}
    &\begin{pmatrix}
        \Lambda_{V_t,k} & \Lambda_{V_tV_b,k}\\-\Lambda_{V_tV_b,k}^\dag & 0
    \end{pmatrix}
     \\
    &\qquad\qquad =
    \begin{pmatrix}
        \Lambda_{V_t} & \Lambda_{V_tV_b}\\-\Lambda_{V_tV_b}^T & 0
    \end{pmatrix}+
    \begin{pmatrix}
        \Lambda_{V_tV_r} & \Lambda_{V_tV_l}\\\Lambda_{V_bV_r} & 0
    \end{pmatrix}
    \frac{1}{
    \begin{pmatrix}
        \Lambda_{V_r} & \Lambda_{V_rV_l}\\-\Lambda_{V_rV_l}^T & 0
    \end{pmatrix}
   - \Lambda^\text{ref}_k }
   \begin{pmatrix}
        \Lambda_{V_tV_r}^T & \Lambda_{V_bV_r}^T\\\Lambda_{V_tV_l}^T & 0
    \end{pmatrix},
\end{aligned}
\end{equation}
where $\Lambda_k^\text{ref}$ is the correlation matrix of the reference state of the $l, r$ bonds in momentum space:
\begin{align}
    \Lambda^\text{ref}_k = \begin{pmatrix}
        0 & -e^{ik}I_{n}\\e^{-ik}I_{n} & 0
    \end{pmatrix}
\end{align}
The reference state is defined over $2n$ Majorana modes in each unit cell (in the Hilbert space of two virtual legs).
This correlation matrix determines the channel in the momentum space, $A_k = \Lambda_{V_t,k}$ and $B_k = \Lambda_{V_tV_b,k}$.

Further explicit calculation leads to a closed form expression for $\Lambda_{V_t,k}$ and $\Lambda_{V_tV_b,k}$, 
\begin{equation}
\begin{aligned}
    \frac{1}{
    \begin{pmatrix}
        \Lambda_{V_r} & \Lambda_{V_rV_l}\\-\Lambda_{V_rV_l}^T & 0
    \end{pmatrix}
   - \Lambda^\text{ref}_k } &= 
   \begin{pmatrix}
       0 & \frac{1}{e^{-ik} - \Lambda_{V_rV_l}^T}\\ -\frac{1}{e^{ik} - \Lambda_{V_rV_l}} &
       \frac{1}{e^{ik} - \Lambda_{V_rV_l}}\Lambda_{V_r}\frac{1}{e^{-ik} - \Lambda_{V_rV_l}^T}
   \end{pmatrix}\equiv \begin{pmatrix}
       0 & Q_k\\-Q_k^\dag & W_k
   \end{pmatrix},
\end{aligned}
\end{equation}
thus
\begin{equation}
\begin{aligned}
    A_k &= \Lambda_{V_t,k} = \Lambda_{V_t} + \Lambda_{V_tV_l}W_k\Lambda_{V_tV_l}^T + \Lambda_{V_tV_r}Q_k\Lambda_{V_tV_l}^T - \Lambda_{V_tV_l}Q_k^\dag \Lambda_{V_tV_r}^T\\
    B_k &= \Lambda_{V_tV_b,k} = \Lambda_{V_tV_b} - \Lambda_{V_tV_l}Q_k^\dag \Lambda_{V_bV_r}^T,
\end{aligned}
\end{equation}
Importantly, $\Lambda_{V_t,k}$ and $\Lambda_{V_tV_b,k}$ are rational trigonometric polynomials; they are analytic functions of $k$ up to at most removable singularities.

\section{Proof of Lemma~\ref{lemma:projectoranalytic}}\label{app:proofLemProjectorAnalytic}

\begin{taggedlemma}{3.1}
The dimension of the preserved space $V_{k,u}$ is a constant in $k$ except for a finite number of momenta.
In particular, the hermitian projector $P_{k,u}$ onto $V_{k,u}$ can be composed into two hermitian projectors $\tilde{P}_{k,u}$ and $P'_{k,u}$, i.e. $P_{k,u} = \tilde{P}_{k,u} + P'_{k,u}$, with the following properties: (1) $P'_{k,u} = 0$ except for a finite number of momenta; (2) $\tilde{P}_{k,u}$ can be extended to an analytic function of $k$ in the vicinity of the real axis.
\end{taggedlemma}
\begin{proof}
The eigenvalue $b_k$ of the matrix $B_k$ satisfies the algebraic equation $\det(B_k- b_kI) = 0$ whose coefficients are analytic on the complex $k$-plane.
Such eigenvalues are also (branches of) analytic functions of $k$ with at most algebraic singularities~\cite{kato2013perturbation}.
Specifically, in the neighborhood of a real momentum $k_0$, the most singular behavior of a set of eigenvalues take the form of the following Puiseux series
\begin{equation}
	b_{k,h} = b_{k_0} + \sum_{n=1}^{+\infty} a_n e^{\frac{2\pi in}{p}h} (k-k_0)^{n/p}\,,\quad p\in\mathbb{Z_+}\,,\, h = 1,\dots,p\,,
\end{equation}
where an eigenvalue $b_{k_0}$ splits into $p$ branches that gives rise to $p$ eigenvalues. 
For nonzero $b_{k_0}$, we have
\begin{align}
    \log(b_{k,h}) = \log(b_{k_0}) + \sum_{n=1}^{+\infty} a'_n e^{\frac{2\pi in}{p}h} (k-k_0)^{n/p}\,.
    \label{eq:series expansion of logb}
\end{align}
We are interested in the behavior of these eigenvalues when $b_{k_0}$ has a unit norm.
Depending on the coefficients $a_n$, there are two situations, one where $b_k$ has a constant norm across an open interval near $k_0$ and the other where $b_k$ has a unit norm\ only at $k_0$. More concretely, we have
\begin{itemize}
\item The expansion \eqref{eq:series expansion of logb} has nonzero real part on the real axis, i.e., $a'_n\neq 0$ for $n/p\not\in \mathbb{Z}$ or $\Re(a'_n)\neq 0$ for $n/p\in \mathbb{Z}$. In this case, $b_{k,h}$ does not have a unit norm near $k_0$

If the leading violation is for a non-integer power $n/p \notin \bbZ$ or an odd integer power, we must have $|b_{k,h}|>1$ for some $h$ and a real $k$; this is forbidden by the condition $B_k^\dag B_k\le I$ for real $k$.
Thus, the leading violation must be at an even integer power, implying that $|b_k| = 1$ only at $k=k_0$ (for real $k$). This type of norm-1 eigenvalue can only appear at discrete momenta.

\item The expansion \eqref{eq:series expansion of logb} is purely imaginary on the real axis. It means that the coefficients should satisfy the condition $a_n' = 0$ for $n/p\not\in \mathbb{Z}$ and $\Re(a'_n) = 0$ for $n/p\in \mathbb{Z}$. Below, we first show that the set of momenta where this condition is satisfied is an open set, and then we show that it is either empty or the entire real axis itself. In this case, the norm-1 eigenvalues $b_k$ are analytic functions of $k$, and they form continuous bands\footnote{Changing the momentum adiabatically from $-\pi$ to $\pi$ may induce a nontrivial permutation of type-I norm-1 eigenvalues. For example $b_{k,\pm} = \pm e^{ik/2}$.}.

To show it is an open set, note that \eqref{eq:series expansion of logb} has a finite convergence radius so that these $b_k$'s have a constant unit norm and no branch cut in an open interval near $k_0$ on the real axis. 
For any other point in this open interval, we can write down a similar Puiseux series expansion and the coefficients should also satisfy the above condition. Therefore, the momenta where this condition is satisfied must form an open set on the real axis. 

To show that the open set is either empty or the entire real axis, we first recall that an open set on the real axis must be a countable union of disjoint open intervals. We show that having an interval $(a,b)$ with finite $a$ or finite $b$ leads to contradiction, implying that the open set is either empty or the entire real axis.
If such an endpoint exists, by definition, it violates the condition for $a_n'$.
Our analysis of the previous situation says that the violation must occur at an even integer power, which contradicts the fact that $|b_k|$ has a unit norm on at least one side of this point. As a result, there cannot be any endpoint of the interval.
\end{itemize}
We define $P'_{k,u}$ and $\tilde{P}_{k,u}$ as the \textit{the sum} of projectors for norm-1 eigenvalues in the two situations, respectively.
By definition, $P'_{k,u}$ is nonzero only at discrete momenta. 
The proof of Lemma~\ref{lemma:akbkeigenvec} can be generalized to $B_k$ in the momentum space to show that $\tilde{P}_{k,u}$ is a Hermitian projector for real $k$; therefore it is analytic near the real axis~\cite{kato2013perturbation}, and $P_{k,u} = \tilde{P}_{k,u} + P'_{k,u}$.
\end{proof}

\section{Entanglement spectrum of chiral states}
\label{app:entanglementspectrum}

In this section, we review the implication of having a nonzero bulk Chern number on the single-particle entanglement spectrum in free fermion systems. 
We show that the Chern number can always be detected either by the discontinuous jump of the entanglement spectrum between $\pm\infty$ or by the number of chiral modes in the entanglement spectrum.
Normally, the discussion on this subject is for ground states of local Hamiltonians~\cite{Fidkowski:2009entanglement}. 
Here we consider a slightly general class of states and do not assume a parent Hamiltonian.
In particular, we discuss tensor network states with algebraically decaying correlation functions.

Consider a two-dimensional translationally invariant free-fermion system that does not necessarily have a definite particle number.
It is then more convenient to use the Majorana operator $c_{l=1,2,\ldots}$ to describe the fermion modes.
Any Slater determinant state $\ket{\Psi}$ is fully characterized by the two-point correlation function, also called the spectral projector, $2P_{kl} = \braket{\Psi|c_l c_k|\Psi}$.
In particular, it defines a spectral Chern number $\nu(P) \in \bbZ$ that is given by the TKNN formula in the momentum space
\begin{equation}
    \nu(P) = \frac{1}{2\pi i} \int \Tr \tilde{P} \big( \frac{\partial \tilde{P}}{\partial q_x} \frac{\partial \tilde{P}}{\partial q_y} - \frac{\partial \tilde{P}}{\partial q_y} \frac{\partial \tilde{P}}{\partial q_x} \big)  dq_x dq_y
    \label{eq:TKNN}
\end{equation}
where $\tilde{P}(q_x,q_y)$ is the Fourier transformation of $P_{kl}$.
In the following, we first review the connection between \eqnref{eq:TKNN} and the number of Majorana edge modes by following Appendix~B of \cite{Kitaev:2005hzj}, and then we explain its implication on the single-particle entanglement spectrum. 

Let us put the system on an $L_x \times L_y$ cylinder with the $x$-direction being periodic and $y$-direction left open.
With translational invariance in the $x$-direction, we consider the spectral projector in each momentum sector
\begin{equation}
    P(q_x) = \sum_{\text{occupied}} \ket{\psi_n(q_x)} \bra{\psi_n(q_x)}\,,\quad q_x \in [0,2\pi)\,,
\end{equation}
where the summation runs over all the occupied single-particle states.
Without loss of generality, we focus on the chiral modes on the upper edge, if exists, by considering the restricted spectral projector $P_A(q_x) \equiv \Pi_A P(q_x) \Pi_A$. Here $\Pi_A$ is a real-space projector for the upper half-cylinder, i.e.
\begin{equation}
    (\Pi_A)_{jl} = f(j) \delta_{jl}\,,\quad f(j) = \begin{cases} 1 & j \in A \\ 0 & j \notin A \end{cases}\,.
\end{equation} 
Suppose there is a single chiral edge mode. As the momentum $q_x$ continuously vary from $0$ to $2\pi$, we expect $P_A(q_x)$ to change discontinuously by $\ket{\psi} \bra{\psi}$ when a single-particle edge state $\ket{\psi}$ is suddenly occupied/unoccupied.
This gives a sudden jump of $\Tr(P_A(q_x))$ by 1.
Since $\Tr(P_A(q_x))$ is a periodic function of $q_x$, the sudden jump must be compensated by a smooth variation.
Physically, the smooth variation originates from the deformation of the single-particle wave functions of the occupied modes.
Based on this argument, the number of the edge modes is given by
\begin{equation}
    \nu_{\text{edge}}(P) = - \int \Tr \Big(\frac{dP_A(q_x)}{dq_x}\Big)_{\text{smooth}} dq_x\,,
    \label{eq:nu edge}
\end{equation}
where we only include the smooth variation of the restricted spectral projector.
One can then follow the manipulation in the Appendix.~B of \cite{Kitaev:2005hzj} to show
\begin{equation}
    \nu(P) = \nu_{\text{edge}}(P)\,,\quad \text{given }|P_{kl}(q_x)| \leq \frac{A}{|k-l|^{\alpha}}\,,\,\alpha > 1\,,
    \label{eq:nu equal nu edge}
\end{equation}
where $|k-l|$ is the distance between the two sites $k$ and $l$, $A$ is a positive constant. 
Namely we have the equality between \eqnref{eq:TKNN} and \eqref{eq:nu edge} if $P(q_x)$ is a quasi-diagonal matrix.
Thus, we can safely use \eqnref{eq:nu edge} as the definition of the spectral Chern number as long as the bulk-correlation function for each $q_x$, $P_{kl}(q_x)$, decays faster than $1/|k-l|$.

To understand the implication of \eqnref{eq:nu edge} on the entanglement spectrum, let us recall that the single-particle entanglement Hamiltonian $K_A$ is related to the restricted spectral projector $P_A$ by
\begin{equation}
    K_A = \log \frac{1 - P_A}{P_A}
\end{equation}
The eigenvalues of $P_A(q_x)$ are bounded between 0 and 1, with the lower and upper bound being mapped to $+\infty$ and $-\infty$ respectively.
The eigenstates of $P_A$ with the zero or unit eigenvalue correspond to either fully empty or fully occupied states, which live at the $\pm \infty$ of the entanglement spectrum and do not contribute to the entanglement.
Suppose that there is a single chiral edge mode that causes a discontinuous change in $P_A$.
In terms of the spectrum of $P_A(q_x)$, the discontinuity manifests as a sudden change in the number of zero and unit eigenvalues. 
The spectrum then must also contain a smooth part that chirally traverses between $0$ and $1$ to compensate the sudden change.
Translating these observations into the entanglement spectrum, the sudden change is mapped to $\pm \infty$ and is invisible, while the smooth part is mapped to a chiral branch (see \figref{Fig:entanglementspec}~(a)). Note that the momentum is defined in a finite Brillouin zone and therefore having a chiral branch implies the presence of discontinuity in the entanglement spectrum.

Tensor network states, especially those with algebraically decaying correlation functions, may not have exponentially localized ``chiral edge modes". Nonetheless, Eq.~\eqref{eq:nu equal nu edge} still holds and these states do have a well-defined Chern number if the correlation function decays fast enough. 
Furthermore, the Chern number can always be detected either by the discontinuous jump of the entanglement spectrum between $\pm\infty$ or by the number of chiral modes in the entanglement spectrum (see \figref{Fig:entanglementspec}~(b)).
In the main text, we prove that the entanglement spectrum in a free-fermion isoTNS must be an analytical function, excluding the possibility of having a chiral branch (see \figref{Fig:entanglementspec}~(c)).

\bibliography{ref.bib}
\end{document}